\documentclass[10pt, a4paper]{article}
\usepackage{typearea}
\typearea{12}
\usepackage{hyperref}
\usepackage{amsfonts,amssymb,mathtools}
\usepackage{amssymb}
\usepackage{amsmath}
\usepackage{graphicx}
\usepackage[english]{babel}
\usepackage[latin1]{inputenc}
\usepackage{verbatim}
\usepackage{enumerate}
\usepackage{tikz}
\usetikzlibrary{backgrounds,calc}
\usepackage{float}
\usepackage{color}
\usepackage{acronym}
\usepackage{algorithmic}
\newcounter{alg}
\usepackage{multirow}
\newtheorem{theorem}{Theorem}
\newtheorem{proposition}[theorem]{Proposition}

\newtheorem{corollary}[theorem]{Corollary}
\newtheorem{lemma}[theorem]{Lemma}

\newtheorem{definition}{Definition}

                      {}

                      {}

\def\squarebox#1{\hbox to #1{\hfill\vbox to #1{\vfill}}}

\def\eod{\vrule height 6pt width 5pt depth 0pt}
\newenvironment{proof}{\noindent {\bf Proof:} \hspace{.677em}}{\hspace*{\fill}{\eod}}

\usepackage{listings}

\newcommand{\pr}[1]{\ensuremath{\text{{\bf Pr}$\left[#1\right]$}}}
\newcommand{\E}[1]{\ensuremath{\text{{\bf E}$\left[#1\right]$}}}

\usepackage{nicefrac}
\usepackage{threeparttable}
\usepackage{booktabs}

\usepackage{caption}
\usepackage{color}
\newcounter{nalg}[section] 
\renewcommand{\thenalg}{\thesection .\arabic{nalg}} 
\DeclareCaptionLabelFormat{algocaption}{Algorithm \thenalg} 

\lstnewenvironment{algorithm}[1][] 
{   
    \refstepcounter{nalg} 
    \captionsetup{labelsep=colon} 
    \lstset{ 
        mathescape=true,
        frame=tB,
        numbers=left, 
        keywordstyle=\color{black}\bfseries\em,
        keywords={,input, output, return, datatype, function, in, if, else, foreach, while, begin, end, for, forall, do, } 
        numbers=left,
        xleftmargin=.04\textwidth,
        #1 
    }
}
{}
\usepackage[inline]{enumitem}

\usepackage{tikz}
\usetikzlibrary{arrows.meta, shapes}

\usepackage{xspace}
\newcommand{\routing}{$\textsc{Alg-Routing}$\xspace}
\newcommand{\sampling}{$\textsc{Alg-Sampling}$\xspace}
\newcommand{\maintainer}{$\textsc{Alg-LDS}$\xspace}
\newcommand{\random}{$\textsc{Alg-Random}$\xspace}
\newcommand{\join}{$\textsc{Join}$\xspace}

\allowdisplaybreaks

\begin{document}

\title{Always be Two Steps Ahead of Your Enemy\footnote{A preliminary version appeared in the proceedings of the 33rd IEEE International Parallel and Distributed Processing Symposium (IPDPS), May 2019.}\newline 
\large{Maintaining a routable overlay under massive churn in networks with an almost up-to-date adversary}}

\author{Thorsten G{\"o}tte\thanks{Department of Computer Science, Paderborn University, Germany, E-mail:thgoette@mail.upb.de. This work is partially supported by the German research council (DFG) in the context of the collaborative research center "\emph{On-the-Fly Computing}" (SFB 901)}\and
Vipin Ravindran Vijayalakshmi\thanks{Chair of Management Science, RWTH Aachen, Germany. E-mail:vipin.rv@oms.rwth-aachen.de. This work is partially supported by the German research council (DFG) Research Training Group 2236 UnRAVeL.}\and
Christian Scheideler\thanks{Department of Computer Science, Paderborn University, Germany, E-mail:scheidel@mail.upb.de}}

\date{}
\maketitle

\begin{abstract}

We investigate the maintenance of overlay networks under massive churn, where an adversary may churn a constant fraction $\alpha n$ of nodes over the course of $O(\log n)$ rounds. 
In particular, the adversary has an almost up-to-date information 
of the network topology as it can observe only a slightly 
outdated topology that is at least $2$ rounds old with a provably minimal restriction that new nodes can only join 
the network via nodes that have taken part in the network 
for at least one round. We show that it is impossible to maintain a connected topology
if adversary has up-to-date information about the nodes' connections.
As our main result we present an algorithm that constructs a new overlay- completely independent of all previous overlays - every $2$ rounds. 
Furthermore, each node sends and receives only $O(\log^3 n)$ messages each round. 
As part of our solution we propose the Linearized DeBruijn Swarm (LDS), 
a highly churn resistant overlay, which will be maintained by the algorithm. 
However, our approaches can be transferred to a variety of classical P2P topologies where nodes are mapped into the $[0,1)$-interval. 
\end{abstract}


\section{Introduction}
Peer-to-peer (P2P) networking has proven to be a useful technique to construct resilient decentralized systems. In a P2P architecture the nodes are connected via the Internet and form a logical network topology,
also known as an overlay network. 
Within the overlay each node has a logical address and logical links that allows it to search and store information in the network.

A key requirement for all applications that rely on P2P networks is reliable
communication between all nodes, i.e., each node should be able to send a message to another node at all times.
This is complicated by the fact that in every large-scale system, errors and attacks are the rule rather than the exception.
At the same time there is usually no or only little admission control for new participants.
This implies a massive amount of churn, i.e., nodes joining and leaving the network at any given time.  
In fact, empirical studies have shown that $50\%$ of all nodes are subjected to churn over the course of an hour~\cite{Stutzbach}.
This alludes for robust and distributed protocols that maintain connected overlays in spite of heavy churn.

In this work we deal with the problem of maintaining a \emph{routable} overlay under adversarial churn. 
We define an overlay as routable, if each node in every round is able to send a message to a given logical address $p \in [0,1)$. In each round the adversary picks a set of nodes that leave the network and proposes a set nodes that join the network.

It is easy to see that an adversary that knows all connections between the nodes can simply partition the network by churning out the neighborhood of a node. 
Previous literature, e.g.,~\cite{AugustineS18,AugustineP0RU15,DreesGS16}, therefore considered models with an additional restriction, where the adversary has slightly outdated information about the nodes' connections. 
In particular, the adversary could access \emph{all} information that is at least $O(\log \log n)$ rounds old, where $n \in \mathbb{Z}^+$ is the minimal number of nodes in the network in each round. 
This includes the nodes' connections, internal
states, random decisions, the contents of all messages, etc.
Within these $O(\log \log n)$ rounds the nodes execute a distributed algorithm that completely rearranges the network topology.
However, we remark that the techniques presented in~\cite{AugustineS18,AugustineP0RU15,DreesGS16} cannot be used 
if one wants to grant the adversary access to even more recent information.

To overcome above mentioned restriction, we propose a trade-off in the form of a $(a,b)$-late omniscient adversary that has an almost up-to-date 
information about the network topology, but is more outdated with regard to all other aspects. 
In particular, it has full knowledge of the topology after $a$ rounds and complete knowledge of messages, internal states, etc. after $b$ rounds.
In the real world, an adversary with similar properties could, e.g., be an agency eavesdropping on an Internet exchange points.
They can see \emph{who} communicates based on the involved IP-addresses but are unable to decrypt the messages (or take longer to decrypt them).  

Our main contribution is a distributed overlay maintenance algorithm that completely rearranges the network every $2$ rounds and can therefore  handle a $(2,O(\log n))$-late adversary. Furthermore, the algorithm allows for routing a message to a logical address $p \in [0,1)$ within $O(\log n)$ rounds. The algorithms are randomized and the results hold with high probability~(\emph{w.h.p.})\footnote{Throughout this paper \emph{w.h.p.} means with probability $\left(1-\frac{1}{n^k}\right)$, where $n$ is the number of nodes and $k$ is a tunable constant}.
The overlay we consider in this work is an extension of the Linearized DeBruijn Graph presented in Richa et al.~\cite{DBLP:conf/sss/RichaSS11} (which by itself is based on the DeBruijn Graph and draws ideas from Naor and Wieder~\cite{DBLP:conf/spaa/NaorW03}) that uses quorums of logarithmic size to send and receive messages. The latter is adapted from Fiat et al.~\cite{FiatSY05}, where the authors use this approach for the Chord-Overlay.
We further present a robust algorithm that minimizes the number of messages sent in every step. 
Our approach uses several structural properties of the overlay as well as a careful analysis
of non-independent events to ensure fast reconfiguration of the network.

\subsection{Model}

We assume that time proceeds in synchronous\footnote{Synchronicity is a standard assumpution in the related work as nodes need to react to the adversary's changes in a timely manner.} rounds and observe a dynamic set of nodes $\mathcal{V} := \big(V_0, V_1, \dots \big)$ such that $V_t$ is the set of nodes in round $t$. 
Each node is identified by a unique and immutable identifier denoted by $ID$. A node $u \in V_t$ can send a message to a node $v \in V_t$ only if it knows the $ID$ of node $v$. 
In a real-world network these $ID$s could, e.g., be the nodes' IP addresses.
This results in series of graphs $\mathcal{G} := \big(G_0, G_1, \dots \big)$ with $G_i = (V_i, E_i)$ and $E_t := \{(u,v) \mid u \text{ sends a message to } v \text{ in round } t\}$. 
Observe that each $G_i$ is a directed graph.
Creating an edge may be compared to sending a UDP (User Datagram Protocol) message to the desired receiver or establishing a TCP connection.
We assume that a node can create edges to $O(\log{n})$ different nodes in each round and can send $O(\text{polylog}\,{n})$ bits via each edge.
Note that throughout this paper we assume that an $ID$ is of size $O(\log{n})$. 

Our model assumes that the set $\mathcal{V}$ is determined by an adversary. 
This implies, in every round $t$ the adversary can propose a set $O_t \subset V_{t-1}$ that leaves the network 
and a set $J_t \subset V_t$ that joins the network, i.e., $V_t := V_{t-1} \setminus O_t \cup J_t$.
In particular, the adversary has to comply to the following rules.
\begin{enumerate}
\item \textbf{Lateness}.

As mentioned in the introduction, we consider the adversary to be $\left(2, O(\log n)\right)$-late omniscient, i.e., the adversary has slightly outdated knowledge of the topology, i.e., the series of graphs $\mathcal{G} := (G_0, G_1, \dots)$ created through the communication between nodes.
In particular, since our adversary is $2$-late in round $t$ the adversary has full knowledge of all
graphs until $G_{t-2}$.
Further, it has \emph{no} knowledge of the nodes' internal states and the contents of messages for $O(\log n)$ rounds, i.e., the adversary learns the content of message sent in round $t$ only in round $t+O(\log n)$.
\item \textbf{Churn Rate}.

For all $V_t\in  \mathcal{V}$ it holds that $|V_t| \in [n,\kappa n]$, where $\kappa \ge 1$ is a small constant. In other words, the number of nodes stays within $\Theta(n)$. 
For a suitable value $T \in O(\log n)$, we assume that $V_{t+T} \cap V_{t} \geq (1-\alpha) n$, where $\alpha \in [0,1)$ is a fixed constant. 
This allows the churn to be $O(n)$ in each round as long as there is a \emph{stable} set of size $\Theta(n)$ that remains in the network for at least $T$ rounds.
\item \textbf{Bootstrap Phase}.

We assume that until a round $B \in O(\log^2 n)$ the adversary is inactive and no churn occurs, also known as the \emph{bootstrap} phase. 
We would like to remark that several other works also assume a bootstrap phase to prepare the random sampling (cf. \cite{AugustineS18,AugustineP0RU15,DreesGS16}).
Only after the conclusion of the bootstrap phase, is the adversary allowed to begin churning nodes in or out of the network.
\item \textbf{Restricted Join}.

Additionally, we assume that a new node $v \in V_t\setminus V_{t-1}$ can only join the network via a node $w \in V_t \cap V_{t-2}$, i.e., the node $v$ joins \emph{via} node $w$
in round $t$. In Section~\ref{sec:impossible} we show that for our purposes this is a necessary condition.
Finally, the number of nodes that join the network via the same node $v \in V_t$ is a constant  $\phi \in O(1)$ \footnote{In principle our algorithm in Section \ref{sec:maintainer} could be extended to tolerate $O(polylog \, n)$ joins per node as in \cite{AugustineS18}. We chose a constant because a higher number of joins would only increase the number of messages by a poly-logarithmic factor and does not introduce further algorithmic challenges.}.
\end{enumerate}

We remark that our model incorporates observations from Stutzbach and Reza~\cite{Stutzbach}
that new nodes join and leave very frequently but there is a (relatively) stable set of older nodes. Note that to the best of our knowledge this is a significantly more \emph{flexible} model compared to other related work.
Given the above mentioned constraints, a round $t$ consists of the following four steps.
\begin{enumerate}
    \item At the beginning of each round the adversary can select a set of nodes $O_t \subset V_{t-1}$ that leave the network in round $t$. These nodes do not receive any messages and leave the network \emph{immediately}.
Further, the adversary may propose a set of nodes $J_t$ that joins the network in round $t$. For each node $v \in J_t$ the adversary selects a \emph{bootstrap} node $w \in V_{t} \setminus J_t$ (satisfying the necessary conditions for Restricted Join) that receives a reference to $v$. 
    \item Next, all nodes that are still in the system receive all messages sent in the previous round.
    Note that this even holds for messages that were sent by nodes that were churned out in the current round.
    In other words, if a node can send out a message, it will be received.
    \item After receiving all messages, a node can perform calculations on its local variables and the received messages.
    \item Finally, each node may send messages to other nodes.
    Recall that sending a message to another node implicitly creates an edge in the graph $G_{t+1}$. Every message sent in round $t$ is received in the round $t+1$. Furthermore, due to the lateness condition these edges can only be \emph{seen} by the adversary at the beginning of the round $t+3$.
\end{enumerate}

\subsection{Related Work}

\begin{table}[!t]
\centering
\renewcommand{\arraystretch}{1.3}
\begin{threeparttable}[b]
\caption{Overview of different models in the literature}
\label{table_example}
\centering
\begin{tabular}{@{}cccc@{}}
\hline
\bfseries Paper & \bfseries Lateness\tnote{a} & \bfseries Churn Rate\tnote{b} & \bfseries Immediate\\
\hline
\cite{AugustineS18} & $\left(O(\log \log n),O(\log \log n)\right)$ & $(\alpha n, O(\log \log n))$ & Yes\\
\cite{DreesGS16} & $\left(O(\log \log n),O(\log \log n)\right)$ & $(n - \frac{n}{\log n}, O(\log \log n))$ & No\tnote{c}\\
\cite{AugustineMMPRU13} & $(O(\log n),O(\log n))$ &$\left(O\left(\frac{n}{\log n}\right), O(\log n)\right)$ & Yes\\
Our work & $(2, O(\log n))$ & $(\alpha n, O(\log n))$ & Yes\\
\hline
\end{tabular}
\begin{tablenotes}
\item [a] An adversary is $(a,b)$-late if it has full knowledge of the topology after $a$ rounds
and complete knowledge after $b$ rounds.
\item [b] The churn rate is $(C,T)$ if the adversary can perform $C$ join/leaves in $T$ rounds.
\item [c] Nodes remain in the network for additional $O(\log\log n)$ rounds.
\end{tablenotes}
\end{threeparttable}
\end{table}

There has been extensive work on analyzing overlay networks under high adversarial churn.
As already mentioned in the introduction, these works had a variety of different model assumptions. See Augustine and Sivasubramaniam~\cite{AugustineS18} for a comprehensive survey on previous results. In the following, we only concentrate on models closely related to ours.

First, there was a series of papers (cf. \cite{Scheideler05,FiatSY05,AwerbuchS07}) that assumed only a subset of nodes is subjected to \emph{adversarial} churn. 
However, these nodes could also act byzantine 
and try to sabotage the overlay's maintenance and the routing by sending corrupted messages.
A general assumption was that up to a constant fraction of nodes would be malicious.
Scheideler~\cite{Scheideler05} present a protocol that spreads these nodes
over the network such that each connected subset of logarithmic size contains a constant 
fraction of non-byzantine nodes. 
Fiat et al. \cite{FiatSY05} build upon this work and present a full overlay maintenance
algorithm that provided a robust Distributed Hash Table. In their approach, each virtual address $p \in [0,1)$
is maintained by a committee of $O(\log n)$ nodes. We will reuse this idea in our work.

In more recent works \emph{all} of the nodes are subjected to adversarial churn and not only a fixed set.
However, these works usually do not consider byzantine behavior.
The adversary in these papers can be described by three properties: The \emph{lateness}, the \emph{churn rate}, and 
if it is \emph{immediate}. We say adversary is $(a,b)$-late if it has full knowledge of the topology after $a$ rounds and complete knowledge of all sent messages, internal states, etc. after $b$ rounds.
The churn rate is $(C,T)$ if the adversary can perform $C$ join/leaves in $T$ rounds.
Last, an adversary is immediate if churned out nodes have to leave the network immediately
and without the possibility to send and receive more messages.
Table \ref{table_example} shows an overview over the different models.
Note that the table is only for comparison as it simplifies some of the models and does not depict all of their respective nuances. However, these simplifications do \emph{not} weaken the adversary.

Augustine et al.~\cite{AugustineP0RU15} present an algorithm that builds and maintains an overlay in the presence of a nearly completely oblivious adversary. Here, the overlay no longer has a fixed structure but is an unstructured expander graph of constant degree. Note that this overlay has no 
virtual addressing. 
However, in \cite{AugustineMMPRU13} the authors present a scheme that allows 
to quickly search for data in these networks.

Further, Drees et al.~\cite{DreesGS16} build a structured expander, a so-called $H_d$-Graph, which is the union of $d$ random rings. Their adversary is not only $O(\log\log n)$-late with regard to communication, it also has access to all nodes' memory and all sent messages after $O(\log\log n)$. Nodes that are churned out in round $t$ may remain in the network until some round $T \in O(t+\log\log n)$. Thus, it is not immediate.

Last, the \textsc{SPARTAN} framework presented in \cite{AugustineS18} probably bears the greatest resemblance with our work.
In \textsc{SPARTAN} the nodes maintain a logical overlay resembling a butterfly network.
To ensure robustness each of the butterfly's virtual nodes is simulated by $O(\log n)$ nodes.
The key difference between our work and \textsc{SPARTAN} is the adversary's lateness. Similar to \cite{DreesGS16}, the \textsc{SPARTAN} framework assumes the adversary to be ($O(\log\log n)$,$O(\log\log n)$)-late, but in return allows the churn to be as high as $\alpha n$ in $O(\log\log n)$ rounds. However, unlike \cite{DreesGS16}, \textsc{SPARTAN} allows the adversary to be immediate.

\subsection{Our Contribution}
In this work we present an algorithm which given a dynamic set of nodes $\mathcal{V} := \big(V_0, V_1, \dots \big)$ chosen by a $(2,O(\log n))$-late adversary, creates a series of graphs $\mathcal{G} := \big(G_0, G_1, \dots \big)$ with $G_i := \left(V_i,E_i\right)$, such that it holds \emph{w.h.p.} that $G_i$ is \emph{routable}, i.e., each node can send a message to a logical address $p \in [0,1)$.

The paper is organized as follows.
\begin{itemize}
\item In Section \ref{sec:debrun} we introduce the Linearized DeBruijn Swarm (LDS). This graph topology is based on the linearized DeBruijn Graph presented by Richa et al. \cite{DBLP:conf/sss/RichaSS11} and the concept of swarms used by Fiat et al.~\cite{FiatSY05} for the Chord overlay network. 
\item In Section \ref{sec:impossible} we show that our model assumptions are necessary 
in order the solve the problem. In particular, we show that any adversary can partition 
a network where nodes can join via nodes that themselves just joined one round ago.
Further, we prove that our model requires the adversary to be at least $1$-late with regard to the topology.
\item In Section \ref{sec:routing} we present a routing algorithm for the LDS, which
optimizes the congestion if we want guaranteed message delivery. In this section, we also define when a dynamic overlay is routable. 
\item Section \ref{sec:maintainer} we present our main contribution, an algorithm that rearranges the graph topology such that it is completely rebuilt every $2$ rounds but still allows routing.
The message complexity is $O(\log^3{n})$ messages per node and round \emph{w.h.p}\footnote{Note that we do \emph{not} seek to optimize message complexity}.
\end{itemize}

\subsection{Definitions and Preliminaries}

In this section we present some definitions and results from
probability theory that we will use in the analysis of our algorithms.
During our analysis we deal with both dependent and independent random variables.
The two following general classes of random variables will prove to be useful.
First, there is Negative Correlation:
\begin{definition}[Negative Correlation, see e.g. {\cite[p. 31]{ScheidelerHabil}}]
A set of random variables $X_1, X_2, \dots, X_n$ are said to negatively correlated if any subset $X_S \subseteq X$ it holds that,
\[
	\E{\prod_{i \in S} X_i} \leq \prod_{i \in S} \E{X_i}.
\]
Here, $S \subseteq [1,n]$ is the set of indices in $X_S$.
\end{definition}
Further, there is the slightly stronger notion of Negative Association:
\begin{definition}[Negative Association~\cite{joag-dev1983, Wajc2017NegativeA}]
A set of random variables $X_1, X_2, \dots, X_n$ are said to negatively associated (NA) if for any two functions $f,g$ both monotonically increasing (or both monotonically decreasing) defined on disjoint subsets of
$X$ it holds that,
\[
	\E{f(X) \cdot g(X)} \leq \E{f(X)}\cdot \E{g(X)}.
\]
\end{definition}
Note that all independent and  hyper-geometric random variables are always NA~\cite{dubhashi1998balls}.
\begin{corollary}[NA implies Negative Correlation~\cite{Wajc2017NegativeA}]
Let $X_1, X_2,\dots, X_n$ be a set of NA random variables. Then for all $X_S \subset X$ such that $i\neq j$ it holds that, \[
	\E{\prod_{i \in S} X_i} \leq \prod_{i \in S} \E{X_i}.
\]
\end{corollary}

The following propositions from Joag-Dev and Proschan~\cite{joag-dev1983} will be extensively used in many of our proofs.
\begin{proposition}[\cite{joag-dev1983,dubhashi1998balls}]\label{NA_prop1}
If $X := (X_1, \dots, X_n)$ and $Y := (Y_1, \dots, Y_m)$ are negatively associated sets of random variables that are mutually independent, then the vector $(X, Y) := (X_1, \dots, X_n, Y_1, \dots, Y_n)$ are also negatively associated.
\end{proposition}

\begin{proposition}[\cite{joag-dev1983,dubhashi1998balls}]\label{NA_prop2}
Let $(X_1, \dots, X_n)$ be negatively associated random variables. For some $k \le n$, let $I_1, \dots, I_k \subseteq [n]$ be disjoint index sets. For $j \in [k]$, let $f_j : \mathbb{R}^{|I_k|} \mapsto \mathbb{R}$ be functions that are all non-decreasing or all non-increasing. Define $Y_j := f_j(X_i: i\in I_j)$. Then the random variables $(Y_1, \dots, Y_k)$ are negatively associated.
\end{proposition}
Further, we make use of the following Chernoff Bounds, which are defined as follows.
\begin{lemma}[Chernoff-Hoeffding Bounds~\cite{dubhashi1998balls, MU05}]\label{chernoff_bound}
Let $X := \sum X_i$ be the sum of negatively correlated random variables with $X_i \in \{0,1\}$.Then it holds that for any $0<\delta<1$,
\[
	\pr{X \geq (1+\delta)\E{X}} \leq e^{-\frac{\delta^2\E{X}}{2}}
\] and  \[ \pr{X \leq (1-\delta)\E{X}} \leq e^{-\frac{\delta^2\E{X}}{3}}.
\]
and for any $\delta \ge 1$, it holds:
\[ \pr{X \geq (1+\delta)\E{X}} \leq e^{-\frac{\delta\E{X}}{3}}.
\]

\end{lemma}
Throughout this work we assume that each node in the network is aware of $n$ and $\kappa$, i.e., the lower and upper bound on the number of nodes currently in the network. 
We make this simplification due to Stutzbach and Reza~\cite{Stutzbach}, that the number of nodes stays \emph{relatively} stable. 
Furthermore, in order to simplify notations, we define $\lambda := 2\log{\kappa n}$ \footnote{For convenience we assume throughout this work that $\lambda$ is an integer}.
We would like to remark that all of our algorithms presented in this manuscript may be adapted to work with 
close estimates of $\lambda$ and $\frac{\lambda}{n}$ using approaches presented in \cite{DBLP:conf/sss/RichaSS11,FiatSY05,KingS04,KingLSY07}.
\subsubsection*{DeBruijn Swarm}
\label{sec:debrun}

\begin{figure}
\centering
\begin{tikzpicture}

\def \n {5}
\def \radius {1.5cm}
\def \margin {8} 

\foreach \s in {1,...,\n}
{
  \node[fill, circle,scale=0.3] at ({360/\n * (\s - 1)}:\radius) {};
  \draw[-] ({360/\n * (\s - 1)}:\radius) 
    arc ({360/\n * (\s - 1)}:{360/\n * (\s)}:\radius);
}

\foreach \s in {0.35,0.15,0.286,0.33, 0.88, 0.44, 0.9,0.1,0.43,0.68,0.75,0.512}{
	\node[fill, circle,scale=0.3] at ({360*\s}:\radius) {};
}
\node[fill, diamond, scale=0.2, blue] at ({140}:\radius) [] (p1){$p$};
\node[fill, circle, scale=0.3, thick] at ({200}:\radius) (p){$v$};
\node[fill, diamond, scale=0.2, red] at ({0}:\radius) [] (p2){$p$};

\begin{scope}[transparency group, opacity = 0.5]
\draw[->, line width=1pt, red] (p)  to[out=0,in=-180] (p2);
\draw[->, line width=1pt, red] (p) to[out=0,in=-45] (p1);
\end{scope}

\draw[|-|, thick, red] ({360/\n * -0.25}:\radius+5) 
    arc ({360/\n * -0.25}:{360/\n * (0.25)}:\radius+5);

\draw[|-|, thick, blue] ({360/\n * -0.15}:\radius+10) 
    arc ({360/\n * -0.15}:{360/\n * (0.15)}:\radius+10) node[midway, right]{$S(\frac{v}{2})$};

\draw[|-|, thick, red] ({360/\n * 1.7}:\radius+5) 
    arc ({360/\n * 1.7}:{360/\n * (2.2)}:\radius+5);

\draw[|-|, thick, blue] ({360/\n * 1.8}:\radius+10) 
    arc ({360/\n * 1.8}:{360/\n * (2.1)}:\radius+10) node[midway, above left]{$S(\frac{v+1}{2})$};
    
\draw[|-|, thick, blue] ({185}:\radius+10) 
    arc ({185}:{215}:\radius+10) node[midway, left]{$S(v)$};
    
\draw[|-|, thick, red] ({180}:\radius+5) 
    arc ({180}:{220}:\radius+5);
    
\end{tikzpicture}
\caption{A node $v$ is connected \emph{each} node in red areas. Note that the swarms $S(v)$, $S(\frac{v}{2})$, and $S(\frac{v+1}{2})$ are real subsets.
These sets are chosen such that they contain $O(\log n)$ nodes \emph{w.h.p}.
Recall that in a classical Debruijn Graph it would only be connected to the nodes
left and right of $v$ i.e., $\frac{v}{2}$ and $\frac{v+1}{2}$.
}
\label{fig:LDS}
\vspace*{-0.5cm}
\end{figure}

We now present our overlay, the \emph{Linearized DeBruijn Swarm} (LDS), which is a combination of a well-analyzed network overlay of low degree, i.e., the Linearized DeBruijn Graph (LDG) presented in \cite{DBLP:conf/sss/RichaSS11,DBLP:conf/sss/FeldmannS17}
with techniques from robust overlays, i.e., the usage of logarithmic quorums that simulate a single node \cite{FiatSY05}. 
Note that the LDG is inspired by but not equivalent to the classical DeBruijn Graph. The notion of swarms were also described in Fiat et al.~\cite{FiatSY05}.

In the remainder of this section we present the LDS's topology and show some of its basic properties.
Each node $v \in V$ chooses a position $p_v \in [0,1)$ \emph{uniformly} and \emph{independently} at random.
Note that for the sake of convenience, the position of a node $v \in V$, we just write $v$ instead of $p_v$. 
It should always be clear from the context if we refer to the node, its \textsc{ID}, or its (current) position. However, we will attempt to make the context sufficiently clear in order to reduce the level of ambiguity. 

All nodes can calculate the distance to another node via the distance function $d: V^2 \rightarrow [0,1)$. 
Given two nodes $v,w \in V$ the distance function $d$ returns the shortest distance (hop counts) between $v$ and $w$ in the $[0,1)$-torus.

Formally, the function $d:V\times V \mapsto [0,1)$ is defined as follows,
\begin{equation}
d(v,w) := \begin{cases}
	|v - w| & \text{if } |v - w| \leq \nicefrac{1}{2}\\
    1 - |v - w| &\text{otherwise.}
	\end{cases}
\end{equation}

Furthermore, the distance function also satisfies the triangular inequality, i.e.,
\begin{equation}
d(v,w) \le d(v,z) + d(z,w) \qquad \text{ for all }z \in V.
\end{equation}

For convenience we introduce the following notions for the relation between two nodes $u,v \in V$.
If $|u-v|\leq \nicefrac{1}{2}$, then $u$ is \emph{left} (clockwise) of $v$ if $u < v$ and \emph{right} (clockwise) otherwise.
For $|u-v| > \nicefrac{1}{2}$ the relation is reversed. 
Further, the set $\langle u,v \rangle \subset V$ contains all nodes which are right of $u$ and also left of $v$. 
Given a node $w$ we say that $u$ is closer to $w$ than $v$ if $d(u,w) < d(v,w)$. 
Last, we call a node $u$ the closest neighbor of $v$ if there exists no other node $u^\prime \in V$ closer to $v$ than $u$, i.e., $u := \underset{k \in V}{\text{argmin}}~d(k,v)$.

In the LDG presented by Richa et al.~\cite{DBLP:conf/sss/RichaSS11}, each node $v$ connects to exactly six other nodes. Namely, the two closest nodes left and right of $p_v$ and the two closest node left and right of the points $\frac{p_v}{2}$ and $\frac{p_v+1}{2}$ respectively. 
We extend this structure such that, each node connects to $O(\log n)$ closest neighbors.
For a given point $p \in [0,1)$ we call $S(p) \subset V$ the \emph{swarm} of $p$. 
It holds that $v \in S(p)$ if and only if $d(v,p) \leq \frac{c\lambda}{n}$.
Here, $c>1$ is a robustness parameter which should be chosen as small as possible.
These swarms (and \emph{not} the nodes) will be the building blocks of our overlay.
We call the swarms $S(p)$ adjacent to $S(p')$ if there is an edge $(v,w)$ between \emph{every} node $v \in S(p)$ and $w \in S(p')$. 
Note that each swarm $S(p)$ spans an interval of length $\frac{2c\lambda}{n}$ as it consists of two intervals of length $\frac{c\lambda}{n}$ to the left and right of $p$ respectively.
Sometimes it will be necessary to distinguish between the nodes in the left and right interval of $p$, so we define $S^L(p) := \{v \in S(p)\mid v\textit{ is left of }p\}$ and $S^R(p) := \{v \in S(p)\mid v\textit{ is right of }p\}$ as the left and right side of $S(p)$.
Given this notion of swarms, we can now formally define the adapted overlay.
Formally the LDS is defined as follows:
\begin{definition}[Linearized DeBruijn Swarm]\label{def_lds}
Let $V \subset [0,1)$ be a set of points with $|V| = n$ and $\lambda := \log{n}$. 
Then, the LDS $G_c := (V,E_L \cup  E_{DB})$ with parameter $c \in \mathbb{N}$
has the following properties:
\begin{itemize}
\item $(v,w) \in E_L \iff$ $w \in V$ and $d(v,w) \leq \frac{2c\lambda}{n}$.
\item $(v,w) \in E_{DB} \iff$ $w \in V$ and $d\left(\frac{v+i}{2},w\right) \leq \frac{3c\lambda}{2n}$ with $i \in \{0,1\}$.
\end{itemize}
\end{definition}

A Linearized DeBruijn Swarm is illustrated in Figure \ref{fig:LDS}. Over the course of this paper we will refer to the edges in $E_L$ as \textit{list} edges, whereas the
edges in $E_{DB}$ as \textit{DeBruijn} edges. 

From Definition~\ref{def_lds} we state the following lemma.
\begin{lemma}[Swarm Property]
\label{lemma:swarm_propery}
	Consider any point $p \in [0,1)$ and its swarm $S(p) \subset V$.
    Then $S(p)$ is adjacent to $S\left(\frac{p}{2}\right)$ and $S\left(\frac{p+1}{2}\right)$.
\end{lemma}
\begin{proof}
Let $p \in [0,1)$ be any point and $v \in S(p)$ be a node in $p$'s swarm. From Definition~\ref{def_lds} this implies that $d(p,v) \leq \frac{c\lambda}{n}$.
We now show that node $v$ has a connection to every node in $S\left(\frac{p}{2}\right)$ and $S\left(\frac{p+1}{2}\right)$.
For the proof we only analyze adjacency to $S\left(\frac{p}{2}\right)$, since the other case is analogous.
We distinguish between the following two scenarios.
\begin{enumerate}
\item If $|p-v| \leq \frac{1}{2}$, then from (1) it holds that,
\begin{equation}
	d\left(\frac{p}{2},\frac{v}{2}\right) := \Big\lvert\frac{p}{2}-\frac{v}{2}\Big\lvert = \frac{1}{2} |p-v| \leq \frac{1}{2} \frac{c\lambda}{n}.
\end{equation}
Let $u$ be any node in $S(\frac{p}{2})$, then $d\left( u, \frac{p}{2}\right) \le \frac{c\lambda}{n}$. Then using (2) and (3) we get,
\begin{align*}
    d\left(u, \frac{v}{2}\right) &\le d\left(u, \frac{p}{2}\right) + d\left(\frac{p}{2}, \frac{v}{2}\right)
    \\&\le \frac{c\lambda}{n} + \frac{c\lambda}{2n}
    \\&=\frac{3c\lambda}{2n}.\tag{4}
\end{align*}
From Definition~\ref{def_lds}, since node $v$ has a DeBruijn edge to each node $w \in V$ with $d\left( \nicefrac{v}{2}, w\right) \le \nicefrac{3c\lambda}{2n}$, then using (4) the lemma follows.
\item Otherwise, $|p-v| > \frac{1}{2}$. 

Observe that this only occurs if either $p \in \left[0, \frac{c\lambda}{n}\right]$ or $p \in \left[1-\frac{c\lambda}{n},1\right)$, i.e., the point $p$ is close to $0$ or $1$ and $v$ lies on the opposite site of the interval.
We distinguish between the following two cases.
\begin{enumerate}
\item If $p \in \left[0, \frac{c\lambda}{n}\right]$, then it also holds that $\frac{p}{2} \in [0, p]$. 
Implies, 
\begin{align*}
d\left(v,\frac{p}{2}\right) \leq d(v,p) \leq \frac{c\lambda}{n}. \tag{5}
\end{align*}

Then, by the triangle inequality and inequality (5), it holds that for every node $u \in S\left(\frac{p}{2}\right)$,
\[d(u,v) \leq d\left(u,\frac{p}{2}\right) + d\left(\frac{p}{2},v\right) \leq \frac{2c\lambda}{n}.\] 
From Definition~\ref{def_lds}, the node $u$ is then a list neighbor of $v$ and the lemma follows.
\item Otherwise, if  $p \in \left[1-\frac{c\lambda}{n},1\right)$ then it holds that $\frac{p}{2} \in \left[\frac{1}{2}-\frac{c\lambda}{2n},\frac{1}{2}\right)$. 
Now consider the distance between $\frac{p}{2}$ and $\frac{v+1}{2}$.
Here, it holds
\[
	d\left(\frac{v+1}{2}, \frac{p}{2}\right) := \Big\lvert\frac{v+1}{2} - \frac{p}{2}\Big\lvert = \frac{1}{2}\vert{v+1} - p\vert. 
\]
Observe that since $v < p$ and $p < v+1$, it holds that $|(1+v)-p|$ is equivalent to $1-|v-p|$.
Therefore, the inequality simplifies to
$$
    d\left(\frac{v+1}{2}, \frac{p}{2}\right) = \frac{1}{2}(1-|v-p|) = \frac{1}{2}d(v,p).
$$
Since $\frac{1}{2}d(v,p) \leq \frac{c\lambda}{2n}$, applying the triangle inequality we get that for every node $u \in S(\nicefrac{p}{2})$,
\begin{align*}
    d\left(u,\frac{v+1}{2}\right) &\le d\left(u,\frac{p}{2}\right) + d\left(\frac{v+1}{2}, \frac{p}{2}\right)
    \\&\le \frac{c\lambda}{n} + \frac{c\lambda}{2n}
    \\&=\frac{3c\lambda}{2n}. \tag{6}
\end{align*}
From Definition~\ref{def_lds}, since node $v$ has a DeBruijn edge to each node $w \in V$ with $d\left( \nicefrac{v+1}{2}, w\right) \le \nicefrac{3c\lambda}{2n}$, then using (6) the lemma follows.
\end{enumerate}
\end{enumerate}
\end{proof}

The next lemma shows that if nodes are assigned to the $[0,1)$-interval uniformly and independently at random, then all swarms have roughly the same size \emph{w.h.p.}
\begin{lemma}[Swarm Size]
\label{lemma:swarm_size}
Let the swarm size to be $\frac{c\lambda}{n}$ with $c \geq 12k$ and assume all nodes pick their positions uniformly and independently at random. 
Then for any point $p \in [0,1)$ it holds that,
\[
    \pr{\frac{1}{2}c\lambda < |S(p)| < 2c\lambda} \geq 1 - \frac{2}{n^k},
\]
where $|S(p)|$ denoted the number nodes in $S(p)$.
\end{lemma}
\begin{proof}
The proof follows from a standard application of the Chernoff bound. 
Given that each node picks its position uniformly and independently at random in $[0,1)$-interval, the probability that a node chooses a point in an interval of length $\frac{c\lambda}{n}$ is exactly $\frac{c\lambda}{n}$\footnote{Note that since $c$ is a constant and $\lambda \in O(\log n)$, therefore for big enough $n$, $\frac{c\lambda}{n} \le 1$.}. 

Let $X$ be a random variable that counts the number of nodes in $S(p)$. For each $v \in V$, let $X_v$ be a binary random variable such that, $X_v=1$, if $v$ picked a position in a swarm $S(p)$ and $0$, otherwise. Clearly, it holds that $X := \sum_{v \in V} X_v$ and $\E{X} = c\lambda$ for every $p$. 
Furthermore, it holds that $X := \sum_{v \in V} X_v$ is the sum of independent random variables.
Hence, using Lemma~\ref{chernoff_bound} it holds that for $c \geq 12k$,
\begin{align*}
\pr{ X \leq \frac{1}{2}\E{X} } \leq e^{-\frac{\E{X}}{12}} \leq e^{-k\lambda} = n^{-k}. \tag{7}
\end{align*}
    
Moreover, it holds that for $c \geq 2k$,
\begin{align*}
\pr{ X \geq 2\E{X} } \leq e^{-\frac{\E{X}}{2}} \leq e^{-k\lambda} = n^{-k}. \tag{8}
\end{align*}
Then, using (7), (8) and the fact that the expected number of nodes in an interval of size $\frac{c\lambda}{n}$ is $c\lambda$, the union bound yields the desired result.

\end{proof} 

\subsubsection*{Routing in a Linearized DeBruijn Graph}
Our routing algorithm described in Section~\ref{sec:routing} is based on the classical LDG's routing algorithm presented in~\cite{DBLP:conf/sss/RichaSS11,DBLP:conf/sss/FeldmannS17}. Before we go into the details of our algorithm,  
we will first recall the classical LDG's routing algorithm.
Routing in the LDG 
works by a bitwise adaption of the target address.
Recall that we assume each node knows $\lambda$.
Therefore, given any destination $p \in [0,1)$, a node can calculate the first $\lambda$ bits $(p_1, \dots, p_\lambda)$ of $p$'s binary representation.
Then, starting with the least significant bit $p_\lambda$, 
the node $v$ sends the message to the node closest to $x_1 := \frac{v+p_\lambda}{2}$.
For this, it uses the corresponding DeBruijn edge.
After that, the message is sent to the node closest to $x_2 := \frac{x_1+p_{\lambda-1}}{2}$.
This goes on until the first bit $p_1$.
Finally, as a consequence of Lemma~\ref{lemma:swarm_size} there are \emph{w.h.p.} only $O(\log{n})$ hops over list edges left to $p$. 

\begin{definition}[Trajectory]\label{def:trajecory}
Let $v \in V$ be a node and $p \in [0,1)$ be an arbitrary point.
Further let $(v_1, \dots, v_{\lambda}) \in \{0,1\}^{\lambda}$ and $(p_1, \dots, p_\lambda) \in \{0,1\}^{\lambda}$ 
be the $\lambda$ most significant bits of $v$ and $p$ respectively. 
Then the trajectory 
$\tau(v,p) := x_0, \dots, x_{\lambda+1} \in [0,1)^{\lambda+2}$ 
is a series of points defined as follows.
\[
	x_i := \begin{cases}
    v&  i=0\\
	(p_{\lambda-i+1}, \dots, p_{\lambda}, v_1, \dots, v_{\lambda-i})&  i\leq\lambda\\
    p&  i=\lambda+1\\
	\end{cases}
\]
\end{definition}
For each point $x_i$ in the trajectory, forward the message to the node closest to it.
Then, forward the message along list edges until it reaches the target.

\begin{figure}
\centering
\begin{tikzpicture}[thick,scale=8]
  \draw[|-|, thick] (0,0) node [below] {$0$} 
  to (.25,0) node [fill, circle,scale=0.3,label=below:$$] (x1){} 
  to (.5,0) node [fill, circle,scale=0.3,label=below:$$] (x0){} 
  to (.625,0) node [fill, circle,scale=0.3,label=below:$$] (x2){} 
  to (.8125,0) node [fill, circle,scale=0.3,label=below:$$] (x3){} 
  to (1,0) node [below] {$1$};
  
  \draw[|-|, blue] (.2,-.01) 
  to (.25,-.01) node [below] {$S(x_1)$}
  to (.3,-.01);
  
  \draw[|-|, blue] (.45,-.01) 
  to (.5,-.01) node [below] {$S(x_0)$}
  to (.55,-.01);

\draw[|-|, blue] (.575,-.01) 
  to (.625,-.01) node [below] {$S(x_2)$}
  to (.675,-.01);

\draw[|-|, blue] (.7625,-.01) 
  to (.8125,-.01) node [below] {$S(x_3)$}
  to (.8625,-.01);

 \begin{scope}[transparency group, opacity = 0.5]
    \path[->, line width=4pt, red] (x0)  edge   [out=90, in=90,line width=4pt, red] (x1);
  \path[->, line width=4pt, red] (x1)  edge   [out=90, in=90, line width=4pt, red] (x2);
  \path[->, line width=4pt, red] (x2)  edge   [bend left=90, line width=4pt, red] (x3);
\end{scope}
  
\end{tikzpicture}
\caption{Example for the first four steps of a trajectory}
\end{figure}

\section{Impossibility Results and Lower Bounds}
\label{sec:impossible}
In this section, we present two fundamental impossibility results for our model.
First, we show that it is impossible to maintain a connected overlay under massive churn
and a $(0,\infty)$-late adversary. This adversary always has up-to-date information about the topology, but is oblivious of everything else, e.g., the sent messages, the nodes' random decisions, etc. 
Second, we show the necessity that new nodes can only join via bootstrap nodes that are in the network for at least $2$ rounds.

We begin with an auxiliary lemma and show that any adversary with a churn rate $(\alpha n, O(\log n))$ can completely exchange the set of nodes within $O(\log n)$ rounds if $\alpha$ is a constant. 
Therefore it simply churns out the nodes in chunks of size $\alpha n$. 
\begin{lemma}
\label{lemma:gone}
Consider any $(a,b)$-late adversary that proposes a series of nodes $\mathcal{V} := V_0, V_1, \dots $ such that for all $V_t \in \mathcal{V}$ it holds that  $|V_t| \in \Theta(n)$ and $|V_t \cap V_{t+T}| \geq (1-\alpha)n$  for  $T \in O(\log n)$ and some constant $\alpha \in (0,1)$.
Then within $O(\alpha^{-1}\log n)$ rounds the adversary can churn out all nodes from $V_0$.
\end{lemma}
\begin{proof}
 For simplification assume that $\alpha := \frac{1}{\beta}$ with $\beta>1$.
and that $\frac{n}{\beta}$ is an integer. If this is not the case, the adversary can add/remove some nodes such that it holds.
    
Now divide the set $V_0$ into $\beta$ disjoint subsets $V_1^1, \dots, V_0^\beta$ of size $\frac{n}{\beta}$ each. The adversary's strategy is as follows: For $i \in 1, \dots, \beta$ churn out $V_0^i$ in round $T \cdot i$ and replace it with a set $\tilde{V}^i$ of the same size. Thus, the following three statements hold:
    \begin{enumerate}
        \item Recall that $\beta$ is a constant. Thus, after $\beta \cdot T$ rounds all nodes $V_0$ have been churned out.
        \item  For each node that is churned out in round $t$ a new node is churned in. This ensures that the number of nodes in each round is $n$.
        \item Within a period of length $T$ only one set of size  $\frac{n}{\beta} = \alpha n$ is churned out. Thus, there always is a subset of size $(1-\alpha)n$ which remains in the system for $T$ rounds.
    \end{enumerate}
    Thus, the strategy fulfills all requirements and churns out all nodes from $V_0$ within $O(\log n)$ rounds. This was to be shown.
\end{proof}

We now show the impossibility for a $(0,\infty)$-late adversary. 
The idea behind this proof is as follows: Consider a node $v \in V_t$ joining the network in round $t$ via some node $w \in V_t$. 
Then only $w$ and all nodes $w$ communicates with know $v$. 
A $0$-late adversary can immediately detect and churn out these nodes. 
Thus, no node in the entire system knows $v$. 
The result is stated in the following lemma. 
\begin{lemma}
\label{lemma:impossible_0late}
A $(0,\infty)$-late adversary with churn rate $(\alpha n, O(\log n))$ for some $\alpha \in (0,1)$, can disconnect any overlay in $O(\log n)$ rounds.
\end{lemma}

\begin{proof}
Let the execution start at round $0$ and let $V_0$ be the initial set of nodes
with $|V_0| := n$. 
Now consider the following strategy: 
\begin{enumerate}
    \item Let a node $v$ join the network in round $0$ via any node in $V_0$.
    \item Further, let a node $w$ join via $v$ in round $2$.
\end{enumerate}
We will show that within $O(\log n)$ rounds,
a $0$-late adversary can separate $w$ from the network.

Let $D_2 \subset V_0$ be the set of all nodes that $v$ communicates with in round $2$.
Note that $|D_2| \in O(\log n)$ 
because we assume that each node 
can only communicate with a logarithmic number of other nodes in one round.
As $v$ is the \emph{only} node that knows $w$ in round 2, it holds 
that $w$ can only be known by nodes from $D' := D_2 \cup \{v\}$ in round $3$.
On the other hand, $v$ may have sent a set of IDs $D_w \subset V_0$ to $w$
in round $2$. These are the only nodes that $w$ knows.

Since $D'$ is of logarithmic size, there exists an $\alpha \in (0,1)$ such that it is well within the permitted churn size per round.
Let all nodes in $D'$ be churned out in round $3$, i.e., \emph{before} they can communicate with any more nodes. This ensures that $w$'s ID cannot be known to any node in the system.
 
Further, $w$ knows only the IDs received from $v$ as it has only communicated with $w$ until now.
Now continue as follows. 
In each round until all nodes from $V_0$ are gone:
\begin{enumerate}
\item Churn out every node $w$ communicates with. 
This ensures that no new node will learn $w$'s ID.
\item Use the strategy given above to churn out as much nodes from $V_0$ 
as possible and churn in the same amount of new nodes.
Note that the total number of nodes does not change.
\end{enumerate}
Using Lemma \ref{lemma:gone} one can easily verify that within $O(\log n)$ rounds all nodes from $V_0$ are gone.
Since all ID that $w$ knows belong to nodes from $V_0$ and no node in the network knows $w$, it is separated from the network. This concludes the proof.
\end{proof}

We continue with the restrictions for the joining nodes. 
The result is stated in the following lemma.
\begin{lemma}
\label{lemma:impossible_handshake}
Let $v \in V$ be a node that joined in round $t$.
Now assume a model where in round $t+1$ a new node $w \in V$ can
join the network via $v$.
Then a $(\infty,\infty)$-late adversary with churn rate $(\alpha n, O(\log n))$ for some $\alpha \in (0,1)$ can disconnect any overlay after $O(\log n)$ rounds.
\end{lemma}

\begin{proof}
Consider a set of nodes $v_1, \dots, v_T, v_{T+1}$ for some $T \in O(\log n)$ that join the network one after another
such that $v_i$ joins via $v_{i-1}$ in round $i$.
Further, for each $v_i$ let $D_i$ be the set $ID$s that it initially receives from $v_{i-1}$.
Let the execution start at round $0$ and let $V_0$ be the initial set of nodes with $|V_0| := n$.

We first show that the adversary can create a situation where a node $v_{T+1}$ that joined the 
network in round $T+1$ only receives IDs of churned out nodes and thereby disconnecting the network. Consider the following strategy: 
Let $V' := v_0, \dots, v_T$ be a set of nodes such that each $v_t \in V'$ with $t>1$ joins the network in round $t$ via node $v_{t-1}$. To be precise, this implies that at the beginning of round $t$, $v_{t-1}$ knows the $ID$ of $v_{t}$ an can send a message to $v_t$. Any of these messages will arrive $t+1$.
Further, each $v_t \in V'$ is churned out in round $t+2$, i.e., immediately after the round $v_{t+1}$ joined. 

Now let $D_t \subset V$ be set of all IDs that $v_t$ knows in round $t+1$.
We now claim that it holds $D_t \subset D_1 \cup \{v_{t-1}\}$ for all $t>1$.

We proof the claim via induction:
\begin{itemize}
    \item For the induction's beginning consider $t=1$ and the corresponding node $v_1$. Here, the claim trivially holds as $D_1 \subseteq D_1$. 
    \item For the induction's step consider any $v_t$ with $t>1$ and assume that the claim holds for $v_{t-1}$, i.e., it holds that $D_{t-1} \subset D_1 \cup \{v_0, \dots, v_{t-2}\}$.
Now consider the join of $v_{t}$ in round $t$.
Any message that reaches $v_{t}$ in round $t+1$ must be sent in round $t$. 
However, in this round \emph{only} $v_{t-1}$ knows $v_{t}$ and therefore only $v_{t-1}$ may share the references with $v_{t}$.
Thus, $D_{t}$ can only be a subset of $D_{t-1} \cup \{v_{t-1}\}$. 
This proves the claim.
\end{itemize}

Therefore, a node $v_{T}$ that joins in round $T \in O(\log n)$ only receive references to nodes from $V_0$ in round $T+1$ (from $v_{T-1}$). 
Note that any further reference can reach $v_T$ only in round $T+2$.

By Lemma \ref{lemma:gone} the adversary may have churned out all nodes from $V_0$ for some $T \in O(\log n)$.
Now let a new node $v_{T+1}$ join via $v_T$ in round $T+1$. 
Then $v_T$ cannot introduce $v_{T+1}$ to any node currently in the network and further cannot introduce any node $w \in V_{T+1}$ to $v_{T+1}$. 
Thus, once $v_T$ is churned out in round $T+2$, it holds that $v_{T+1}$ is isolated from $V_{T+2}$.
This was to be shown.
\end{proof}

We would like to remark that this impossibility is different from the similar statement in \cite{AugustineP0RU15} because we allow a node to communicate with $O(\log n)$ different nodes instead of constantly many.

\section{Routing and Sampling in the LDS under Churn}
\label{sec:routing}

\begin{figure}
\begin{algorithm}[caption={\routing}, label={alg:routing}]
$\textbf{Desc:}$ This algorithm is executed on a series of routable graphs $\mathcal{D} := (D_1, H_1, \dots)$. The algorithm routes a message $m$ from any node $u \in V_1$ to swarm $S_{2\lambda+2}(p)$. 

$\textbf{Note:}$ The following code is executed by each node $u \in V^t$ every round $t$. $\emph{W.l.o.g.}$ the Forwarding step executed $\text{in}$ even rounds and the Handover step is executed $\text{in}$ odd rounds. The initial step may be executed $\text{in}$ odd or even rounds. Messages are delivered to their target swarms $\text{in}$ even or odd round respectively.

$\rule[1pt]{3.25cm}{1pt}$ $\textbf{Initial step}$ $\rule[1pt]{3.5cm}{1pt}$
$\textbf{Upon sending a message } $m$ \textbf{ to}$ $p$
  $(d_1, \dots, d_{\lambda}) \longleftarrow$ $\lambda$ most significant bits of $p$
  if $t$ is even:
    Send $\big(p,0,(d_1, \dots, d_{\lambda}),m,t\big)$ to all $w \in S_t(v)$
  else:
    Send $\big(p,0,(d_1, \dots, d_{\lambda}),m,t\big)$ to all $w \in S_{t+1}(v)$
    
$\rule[1pt]{3.05cm}{1pt}$ $\textbf{Forwarding Step}$ $\rule[1pt]{3.05cm}{1pt}$
$\textbf{Upon receiving}$ $m := \big(p,k,(d_1, \dots, d_{\lambda_v}),m,t\big)$
  if $k \leq \lambda_v$:
      $x \longleftarrow \frac{v+d_k}{2}$
      $(w_1, \dots, w_r) \longleftarrow$ $r$ nodes chosen u.i.r from $S(x)$
      Forward $m' :=  \big(p,k+1,(d_1, \dots, d_{\lambda_v}),\lambda_v\big)$ to all $w_i$
  else if $t$ is even:
      Deliver $m$ all nodes $w \in S_{t+1}(p)$ in next round
  else:
      Deliver $m$ all nodes $w \in S_{t}(p)$
      
$\rule[1pt]{3.05cm}{1pt}$ $\textbf{Handover Step}$ $\rule[1pt]{3.05cm}{1pt}$
$\textbf{Upon switching to }D_{i+1}$
  for each message $m$
      $(w_1, \dots, w_C) \longleftarrow$ $r$ nodes chosen u.i.r from $S_{t+1}(x)$
      Forward $m$ to all $w_i$
\end{algorithm}

\end{figure}

In this section, we present a low-congestion routing algorithm \routing.
The algorithm delivers \emph{each} message \emph{w.h.p.} even in the presence of churn and a changing communication structure. 
We also present a sampling algorithm \sampling that allows each node to send a message 
to a uniformly picked random node. 
The underlying technique is adapted from King et al. \cite{KingS04,KingLSY07}.


Our algorithm must perform routing over a dynamic series of graphs $\mathcal{D} := (D_1,D_2, \dots)$
where each $D_i$ is a LDS. However, there are two problems we need to address, i.e., the \emph{churn orchestrated by the adversary} and the \emph{dynamic reconfiguration of the overlays}. The obvious solution would be to send the message not only to the closest node of each trajectory point but to the whole swarm. However, observe that this trivial adaption to the LDG routing algorithm fails in the presence of churn. 
Given that any node on a message's trajectory can be churned out, a fraction of routing requests may never reach their destinations. 
In particular, if the adversary is aware of the topology, it could even churn out the whole swarm for a given trajectory point.
Therefore, we introduce the notion of a \emph{good} swarm adapted from Fiat et al.~\cite{FiatSY05}. 
In their work, a swarm is good if at least a fixed constant fraction of its nodes take part in the next round and hence refer to such nodes as \textit{good}. 

Here, we need a slightly stronger notion as we require the good nodes to be somewhat well spread over the swarm to enable our fast construction.
To be precise, we say the left (right) side of a swarm is good if a constant fraction of its nodes is good and the full swarm is good, if both its left and right side are good.
\color{black}
Further, a LDS is good if all its swarms are good.
This property implies that there is always at least a constant fraction of good nodes in each swarm that can forward the message. 

Besides the churn there is the problem of the dynamically rearranging overlay.
In particular, the main algorithm we later introduce in Section~\ref{sec:maintainer} will create a series of overlays $D_1, D_2, \dots$ which will persist for only $2$ rounds each. That means a node changes its position every $2$ rounds. If now every node would keep all its routing messages and forward them from its new position, 
they would lose all the progress they made so far.
Therefore, we define the so-called \emph{handover} procedure using a helper graph $H_i$.
For any point $p \in [0,1)$, let $S_i(p)$ be the swarm of $p$ in $D_i$ and $S_{i+1}(p)$ be the swarm of $p$ in $D_{i+1}$. Denote by $|S_i(p)|$ the number of nodes in the swarm $S_i(p)$.
We assume that during the change from $D_i$ to $D_{i+1}$ each node from $S_i(p)$ can send a message to any set of nodes from $S_{i+1}(p)$, i.e., the nodes from a helper graph $H_i$ where the swarms $S_{i}(p)$ and $S_{i+1}(p)$ are adjacent. 
Formally, it is defined as follows:
\begin{definition}[Handover Graph]
\label{def:handover_graph}
Let $D_0, D_1, \dots $ be a series of LDS with $D_i = (V_i, E_i)$. 
Then the helper graph $H_i := (V_i \cap V_{i+1}, E^H_i)$ is defined as follows:
\[
    (v,w) \in E_i^H \Longleftrightarrow \exists~p \in [0,1) : v \in S_i(p) \wedge w \in S_{i+1}(p) 
\]
\end{definition}
Given this definition, we can easily see that the following holds:

\begin{lemma}[Handover Property]
\label{def:handover_graph_aux}
Let $p \in [0,1)$ be an arbitrary point in the $[0,1]$-interval and let $S_{i+1}(p)$ 
be its swarm in $D_{i+1}$. 
Then in $H_i$ every node in $S_{i}^r(p)$ knows every other node in $S_{i+1}^r(p)$.
The same holds for the left side.
\end{lemma}
\begin{proof}
As we will see, the property follows almost directly from the definition of swarms and the handover graph.
We will prove the lemma only for the right side as the proof for the left side is completely analogous.
Let $v$ be any node in $S^R_{i+1}(p)$.
By the definition of the Handover graph every node in $S_i(p_v)$ knows $v$ as clearly $v \in S_{i+1}(p_v)$.
Since $v$ is right of $p$ and within distance $\frac{c\lambda}{n}$ of $p$ as it is in $p's$ swarm,
it holds that $S^R_i(p) \subset S_i(v)$.
Thus, by combining it with the definition of the Handover Graph, we get that every node $S^R_i(p) \subset S_i(v)$ must know $v$.
Since this holds for all nodes $v \in S^R_{i+1}(p)$, all nodes in $S_i^R(p)$ must know all nodes in $S_{i+1}^R(p)$ and the lemma follows.
\end{proof}
\color{black}

Therefore, the switch can (almost) be handled like every other routing step from one swarm to another. The only difference is that we only send messages from and to the left and right side of each swarm respectively. 
However, if we choose the swarm size $\frac{c\lambda}{n}$ big enough, this makes no difference.

Later, in Section \ref{sec:maintainer} we will see how to implement such a handover graph whereas here we just treat it as a property for a simpler description. 
Note that we call a helper graph $H_i$ good if for each $p \in [0,1)$ a $\nicefrac{3}{4}$-fraction of all nodes in $S_{i+1}(p)$ is not churned out in the next round.

We summarize our observations in following definition for a \emph{routable} series of graphs: 
\begin{definition}[Routable Graphs]
Let $\mathcal{D} := (D_0, H_0, D_1, H_1 \dots)$ be series of graphs defined on nodes $\mathcal{V}:=(V_0,V_1, \dots)$, s.t., each $D_i$ consists of nodes in $V_{2i}$.
Then we call $\mathcal{D}$ routable, if 
\begin{enumerate}
\item each $D_i$ is a LDS,
\item each $H_i$ enables is a handover from each $D_i$ to $D_{i+1}$, and
\item each $D_i$ (and $H_i$) is good, i.e., it holds $|S_i(p) \cap V_{2i+1} | \geq \nicefrac{3}{4} \cdot |S_i(p)|$ for all $p \in [0,1)$.
\end{enumerate}  
\end{definition}

\subsection{The Routing Algorithm}

We now present a routing algorithm \routing for a dynamic series of \emph{routable} graphs $\mathcal{D} := (D_1,H_1,D_2,H_2, \dots)$. 
A trivial extension of the LDG routing algorithm would send each message to the whole swarm of each trajectory point.
However, forwarding a message to a whole swarm would require $O(\log^2 n)$ messages to be sent in each step.
In order to limit this to $O(\log n)$ messages\footnote{Given that $\alpha n$ nodes may fail in single round and we want to route each message on the first try, it reasonable that
one needs $\mathcal{O}(\log n)$ copies of a message each round to ensure the survival of at least one $\emph{w.h.p.}$}, we adapt the approach as follows.
Assume a node $v \in V_t$ wants to route a message $m$.
We first forward $m$ to all nodes in $S(v)$.
Then, each node in $S(v)$ picks $r \in \Theta(1)$ nodes uniformly and independently at random from the next swarm $S(x_{1})$ in the trajectory and forwards $m$ to them. 
Then, each node that received $m$ at least once, forwards it to $r$ nodes in $S(x_{2})$ and 
so on. 
Only in the last step the message is forwarded to all nodes of the target swarm
to ensure that the whole swarm receives the message. Listing \ref{alg:routing} depicts the pseudo-code for \routing.

\subsubsection*{Analysis}
In this section, we analyze \routing. In particular, we prove that \emph{w.h.p.} all messages reach their target and further analyze the \emph{dilation}, i.e., the number of steps until a message reaches its target, and the \emph {congestion}, i.e., the number of messages handled by each node in a round.
Note that the latter depends on how many messages are sent each round and how their
destinations are chosen. We would like to remark that we assume that each node sends exactly the same number of messages and chooses their destinations independently and uniformly at random. 

\begin{theorem}
\label{thm:routing}
Let $\mathcal{D}$ be a routable series of LDS defined on nodes $\mathcal{V}:=(V_1,V_2, \dots)$.
Further, let each node $v \in V_1$ start $t \in \mathbb{Z}_+$ messages to random targets $p \in [0,1)$.
Then \routing with $r\ge 16$ delivers each message with dilation exactly $2\lambda+2$ and congestion $O(t\log n)$ w.h.p.
\end{theorem}
\begin{proof}
We begin the proof with the following lemma where we show that each message reaches its target swarm after exactly $2\lambda+2$ rounds if it is not churned out.

\begin{lemma}
\label{lemma:routing_correctness}
Let $\mathcal{D}$ be a routable series of LDS defined on nodes $\mathcal{V}$. Let $v$ be any node in $V_t \in \mathcal{V}$, which sends a message to point $p \in [0,1)$ along the trajectory $\tau(v,p)$ using $\text{\routing}$. Then it holds that the message arrives at a node in $S_{\lambda+1+t}(p)$ in exactly $2\lambda +2$ rounds.
\end{lemma}

\begin{proof}
The proof follows by an induction over the trajectory $i = 0, \dots, \lambda$ and the fact that by our choice of $\lambda$ the points $x_{\lambda}$ and $p$ are close by.

For the induction, observe that in each even step $j$ the message is forwarded along the trajectory and therefore moves from $S_j(x_{j-1})$ to $S_j(x_{j})$, whereas in each odd step the message is handed over and therefore moves from $S_j(x_{j})$ to $S_{j+1}(x_{j})$). 
Lemma \ref{lemma:swarm_propery} and the handover property imply that the nodes have necessary connections for each step but the last. 
We now prove by induction that, for each $i \in [\lambda]$, that each copy of the message is stored at a node $v \in S_i(x_i)$ in round $2i+t$. W.l.o.g. assume the message is initiated in round $t=0$.

\begin{itemize}
    \item[\textbf{(IB)}] Consider round $i=0$, i.e., the round in which the message is started.
    In this round, the message is at $v=x_0$ and therefore, is known by $S_{0}(x_0)$ in the first step. 
    \item[\textbf{(IS)}] Now suppose that the induction hypothesis holds for any arbitrary $i\in[0,\lambda-1]$. 
We get that in round $2i$ any copy of the message is at a node in $S_{i}(x_{i})$. 
Now since the round $2i$ is an even round, $\text{\routing}$ performs a forwarding step in $D_i$ along the trajectory, i.e., every copy of the message is sent to some node in $S_i(x_{i+1})$. Observe that, the swarm property in Lemma~\ref{lemma:swarm_propery} ensures that each node in $S_{i}(x_{i})$ has a connection to $S_{i}(\frac{x_{i}}{2})$ and $S_{\lambda}(\frac{x_{i}+1}{2})$. 
Now since $x_{i+1}$ is either $\frac{x_{\lambda}}{2}$ or $\frac{x_{\lambda}+1}{2}$, it holds that each node $S_{i}(x_{i})$ has an edge to \emph{each} node in $S_{i}(x_{i+1})$. 
Therefore, every copy of the message can be successfully forwarded to each node in $S_{i}(x_{i+1})$. 
Next, in round $2i+1$, $\text{\routing}$ performs a handover operation on every copy of the message in $S_{i}(x_{i+1})$ (overlay $H_i$). 
Now we use the Handover Property and observe that each node in $S_{i}(x_{i+1})$ has by the definition of $H_{i}$, an edge to \emph{each} node in $S_{i+1}(x_{i+1})$. 
Therefore, every copy of the message can be successfully be sent to a node in $S_{\lambda+1}(x_{\lambda+1})$ and therefore, available in round $2i+2 := 2(i+1)$. 
This concludes the induction step.
\end{itemize}
The induction above implies that the message is known by all nodes in $S_{\lambda}(x_{\lambda})$ in round $2\lambda$.
Now recall that $x_{\lambda}$ and $p$ are equal in their first $\lambda$ bits.
This implies that the distance between $x_{\lambda}$ and $p$ is at most 
\begin{align*}
    d(x_{\lambda}, p) &\leq \frac{1}{2^{\lambda}}\\
    &= \frac{1}{e^{\log(2)\lambda}} = \frac{1}{\left(e^{\lambda}\right)^{\log(2)}}\\
    &=\frac{1}{\left({\kappa n}\right)^{2\log(2)}}\\
    &\leq\frac{1}{\kappa n} \leq \frac{1}{n},\\
\end{align*}
as we defined $\lambda := 2\log(\kappa n)$.
Therefore, the swarms $S_{\lambda}(x_{\lambda})$ and $S_{\lambda}(p)$ are adjacent and the message can be forwarded and handed over as described in the induction step.
This proves the lemma.

\end{proof}

In the proof of Lemma~\ref{lemma:routing_correctness} we omitted the fact that not all nodes of a swarm forward the message as
they may be churned out before they can do so. 
Of course, if a complete swarm is churned out, the message surely can't be forwarded along the trajectory. However, as stated earlier we assume that all swarms are \emph{good}, i.e., only a constant fraction of each swarm is malicious and does not forward the message. 

\begin{lemma}
\label{lemma:good_nodes}
Consider a set $S$ with $|S| \geq \frac{c\lambda}{2}$ nodes picked uniformly at random from the set of all nodes in any given round. Let $G\subset S$ be the set of good nodes in $S$, Then it holds that for churn parameters $\alpha=\frac{1}{16}$ and $\kappa = \left(1 + \frac{1}{16}\right)$, with $c \geq 510k$,
\[
    	\pr{ |G| \leq \frac{14}{17}|S|} \leq \frac{1}{n^k}.
\]
\end{lemma}
\begin{proof}
Observe that by definition, for a churn rate of $\alpha$, there are at least $(1 - \alpha)n$ nodes that would survive into the next round. Therefore, there are at least $\frac{15}{16}n$ good nodes in any given round. Also, since there could be at most $\kappa n$ nodes, there are at most $\frac{17}{16}n$ nodes in any given round. Therefore, the fraction of good nodes in the $[0,1)$-interval in any given round is then at least $\frac{15}{17}$. Furthermore, nodes pick their position uniformly at random in the $[0,1)$-interval.

Let $S \subset V_t$ be a set of nodes in round $t$ picked uniformly at random from the $[0,1)$-interval such that, $|S| \ge \frac{c\lambda}{2}$. For each $v \in S$, let $X_v$  be a $\{0,1\}$ r.v. such that $X_v = 1$ if $v$ is a good node, and $X_v=0$ otherwise. We know that, \[\pr{X_v = 1} = \frac{15}{17}.\]

Therefore,

\[\E{X_v} = \frac{15}{17}.\]
The expected number of good nodes in the set $S$ is then,
\[\E{\sum_{v \in S}X_v} = \frac{15}{17} |S|.\]

The random variables $X_1, \dots, X_{|S|}$ represent random sampling without replacement and therefore observe a permutation distribution which are negatively associated~\cite{joag-dev1983}. 

Now, let $X_S = \sum_{v \in S}X_v$ be the number of good nodes in the set $S$. Then applying the Chernoff bound on NA random variables with $c \ge 510\cdot k$ we get,
\begin{align*}
\pr{ X_S \leq \left(1 - \frac{1}{15}\right)\E{X_S} } \leq \exp\left(-\frac{|S|}{255}\right) \leq \exp\left(-\frac{c\lambda}{510}\right) \leq \exp\left(-k\lambda\right) = n^{-k}.
\end{align*}
\end{proof}

Thus, as long as we observe only $O(n^{k})$ swarms and each side of each swarm has more than $\frac{c\lambda}{2}$ nodes w.h.p, a simple union bound implies that in all swarms both the left and right side are good w.h.p. 

Using Lemma~\ref{lemma:good_nodes} we can now show that the messages reach their destination w.h.p.
\begin{lemma}
\label{lemma:routing_whp}
Let $m$ be a message that is routed along $\tau(v,p) := x_0, \dots, x_{\lambda+1}$ using \routing.
Then, it holds that all nodes in $S_{\lambda+1}(x_{\lambda+1})$ receive $m$ w.h.p after exactly $2\lambda+2$ rounds for $r = 16$.
\end{lemma}
\begin{proof}
We prove by induction that, for each $t\in [2\lambda+2]$, it holds that if $t$ is an even round then, at least half of all good nodes in $S_{t/2}(x_{t/2})$ receive the message $m$ w.h.p. Otherwise, if $t$ is an odd round, at least half of all good nodes in $S^R_{\lceil t/2 \rceil}(x_{\lceil t/2 \rceil - 1})$ and $S^L_{\lceil t/2 \rceil}(x_{\lceil t/2 \rceil - 1})$ receive the message $m$ w.h.p.
\begin{itemize}
\item[\textbf{(IB)}] Consider round $t=0$, i.e., the round in which the message is initiated. If this is an even round, then observe that $\routing$ forwards message $m$ from node $v$, i.e., $x_0$, to \emph{all} nodes in $S_0(x_0)$. 
Using Lemma~\ref{lemma:swarm_size} and~\ref{lemma:good_nodes} we can conclude that at least half of all good nodes in the swarm $S_0(x_0)$ received the message $m$ and survive until the next round. 
Therefore, the induction hypothesis holds. 
Otherwise, if $t=0$ is an odd round, $\routing$ handovers the message $m$ from node $v$ to all nodes in $S_1(x_0)$. 
Again, using Lemma~\ref{lemma:swarm_size} and~\ref{lemma:good_nodes}, $S_1(x_0)$ is a good swarm and therefore, the induction hypothesis holds. 
\item[\textbf{(IS)}] Now suppose the induction hypothesis holds for any arbitrary $t\in [2\lambda+1]$ and w.l.o.g. assume that round $t = 2\lambda+1$ is an odd round. 

Note that the handover step works analogously, the only difference is that we do not send messages from one swarm to another but from the right/left side from one swarm to the right/left side of another.
However, since these sets are smaller, we assume the swarm size to be at least $\frac{c\lambda}{n}$ (and not $\frac{2c\lambda}{n}$).
This way, we cover both cases.
\color{black}
By the induction's hypotheses, at least half of all good nodes in $S_{\lambda+1}(x_\lambda)$ received message $m$ w.h.p. and therefore, each of these nodes forward $r$ copies of the message $m$ to nodes picked uniformly and independently at random from the swarm $S_{\lambda+1}(x_{\lambda+1})$. We now show that at least half of all good nodes in the swarm $S_{\lambda+1}(x_{\lambda+1})$ receive the message w.h.p. in the round $2\lambda+2$.
~

From Lemma \ref{lemma:swarm_size}, it holds that $|S_{\lambda+1}(x_{\lambda+1})| \leq 4|S_{\lambda+1}(x_\lambda)|$ w.h.p. Furthermore, by the induction hypothesis we know that at least half of all good nodes in $S_{\lambda+1}(x_\lambda)$
received $m$ and therefore forward $r$ copies of $m$. 
Thus, in total there are at least $\nicefrac{r}{8}\cdot |S_{\lambda+1}(x_{\lambda+1})|$ 
copies of $m$ sent to $S_{\lambda+1}(x_{\lambda+1})$. 
The probability that any of these messages is sent to a given node $v^\prime \in S_{\lambda+1}(x_{\lambda+1})$ 
is $\frac{1}{|S_{\lambda+1}(x_{\lambda+1})|}$, since the destinations are chosen uniformly at random.
Observe that one can view the forwarding of messages from $S_{\lambda+1}(x_{\lambda})$ to uniformly and independently picked nodes in $S_{\lambda+1}(x_{\lambda+1})$ as a balls-into-bins experiment, where (at least)
$\nicefrac{r}{8}\cdot |S_{\lambda+1}(x_{\lambda+1})|$ balls are thrown into $|S_{\lambda+1}(x_{\lambda+1})|$ bins. 
Using Propositions~\ref{NA_prop1} and~\ref{NA_prop2} one can show that the number of nodes in $S_{\lambda+1}(x_{\lambda+1})$ that receive at least one ball is NA~\cite{dubhashi1998balls}. 
For $i \in \{1, \dots, |S_{\lambda+1}(x_{\lambda+1})|\}$ and $j \in \{1, \dots, \nicefrac{r}{8}\cdot |S_{\lambda+1}(x_{\lambda+1})|\}$, let $X_{i,j}$ be an indicator random variable such that, $X_{i,j} = 1$ if message $j$ is sent to node $i$ (picked uniformly and independently at random) by \routing in round $2\lambda+1$, and $X_{i,j} = 0$ otherwise.

\begin{lemma}[Zero-One Lemma~\cite{dubhashi1998balls}]\label{z1_lemma}
If $Y_1, \dots, Y_n$ are zero-one random variables such that $\sum_i Y_i = 1$, then $Y_1, \dots, Y_n$ are NA.
\end{lemma}

For any fixed $j \in \{1, \dots, \nicefrac{r}{8}\cdot |S_{\lambda+1}(x_{\lambda+1})|\}$, let $Y_i := X_{i,j}$ for all $i \in \{1, \dots, |S_{\lambda+1}(x_{\lambda+1})|$. Then, from Lemma~\ref{z1_lemma} we know that the random variables $Y_1, \dots, Y_{|S_{\lambda+1}(x_{\lambda+1})|}$ are NA. Since each message $j \in \{1, \dots, \nicefrac{r}{8}\cdot |S_{\lambda+1}(x_{\lambda+1})|\}$ is destined to a node that is picked uniformly and independently at random, using Proposition~\ref{NA_prop1}, we can conclude that the set of random variables $(X_{i,j})_{ i \in \{1, \dots, |S_{\lambda+1}(x_{\lambda+1})|\}, j \in \{1, \dots, \nicefrac{r}{8}\cdot |S_{\lambda+1}(x_{\lambda+1})|\}}$ are NA.

For each $i \in \{1, \dots, |S_{\lambda+1}(x_{\lambda+1})|\}$ consider a non-decreasing function as follows,
\[
  X_i = 
  \begin{cases}
    1 & \sum_{j \in \{1, \dots, \nicefrac{r}{8}\cdot |S_{\lambda+1}(x_{\lambda+1})|} X_{ij} > 0\\
    0 & \text{otherwise.}
  \end{cases}
\]

Therefore, for each $i \in \{1, \dots, |S_{\lambda+1}(x_{\lambda+1})|\}$,
\begin{align*}
    \pr{X_i = 1} &= 1- \left(1 -\frac{1}{\mid S_{\lambda+1}(x_{\lambda+1})\mid}\right)^{|S_{\lambda+1}(x_{\lambda+1})|\cdot \nicefrac{r}{8}} 
    \\&\geq 1- \left(\frac{1}{e}\right)^{\frac{r}{8}}.
\end{align*}

From Proposition~\ref{NA_prop2}, we know that the random variables $X_1, \dots, X_{|S_{\lambda+1}(x_{\lambda+1})|}$ are NA. 

~

Let $G_{\lambda+1}(x_{\lambda+1}) \subseteq S_{\lambda+1}(x_{\lambda+1})$ denote the set of good nodes in $S_{\lambda+1}(x_{\lambda+1})$.
For each $v \in G_{\lambda+1}(x_{\lambda+1})$, $G_v$ be a $\{0,1\}$ random variable such that, $G_v = 1$ if node $v$ received at least one copy of $m$ from some node in $S_{\lambda+1}(x_{\lambda})$ and $G_v = 0$, otherwise.
Observe that $X_i$ denotes if a node in $S_{\lambda+1}(x_{\lambda+1})$ received at least one message in round $2\lambda+2$ when exactly $|S_{\lambda+1}(x_{\lambda+1})|\cdot \nicefrac{r}{8}$ messages are sent to $S_{\lambda+1}(x_{\lambda+1})$ uniformly and independently at random. Since $|S_{\lambda+1}(x_{\lambda+1})|\cdot \nicefrac{r}{8}$ is a lower bound on the number of message being sent, for each $v \in G_{\lambda+1}(x_\lambda+1)$ we have that,
\[\pr{ G_v = 1} \ge \pr{X_v =1}. \] 

~
Moreover, the random variables $(G_v)_{v \in G_{\lambda+1}(x_{\lambda+1})}$ are also NA as they can be seen as a subset of $(X_v)_{v \in S_{\lambda+1}(x_{\lambda+1})}$.
To see this, first recall that the random variables $(X_i)_{i \in \{1, \dots, |S_{\lambda+1}(x_{\lambda+1})|\}}$ are NA. In particular, this fact is independent of the number of messages and nodes in $S_{\lambda+1}(x_{\lambda+1})$. Now observe that for each $v \in G_{\lambda+1}(x_{\lambda+1})$, the random variable $G_v$ is a non-decreasing function of its associated random variable $X_v$. Therefore, we can conclude that the random variables $(G_v)_{v \in G_{\lambda+1}(x_{\lambda+1})}$ are also NA.

Let $G := \sum_{v \in G_{\lambda+1}(x_{\lambda+1})} G_v$ be a random variable that counts the number of good nodes in $S_{\lambda+1}(x_{\lambda+1})$ that received at least one copy of $m$ in round $2\lambda+2$. From Lemma~\ref{lemma:good_nodes}, we know that there are $\nicefrac{14}{17}$ fraction of good nodes in $S_{\lambda+1}(x_{\lambda+1})$ w.h.p.
Therefore, the expected number of good nodes that receive at least one message is given by,
\begin{align*}
\E{G} &= \sum_{v \in G_{\lambda+1}(x_{\lambda+1})} \E{G_v} 
\\&\geq \sum_{v \in G_{\lambda+1}(x_{\lambda+1})} \left(1-\frac{1}{e^{\nicefrac{r}{8}}}\right)
\\&= \left(1-\frac{1}{e^{\nicefrac{r}{8}}}\right)\frac{14}{17}|S_{\lambda+1}(x_{\lambda+1})|.
\end{align*}

To complete the induction it suffices to show that,
\[ \pr{G \le \frac{1}{2}|S_{\lambda+1}(x_{\lambda+1})|} \le \frac{1}{n^k}.\]

Using $\delta=\frac{2}{7}$ in Lemma~\ref{chernoff_bound} with $r=16$ we get,
\begin{align*}
    \pr{G \le \left(1 - \frac{2}{7}\right)\left(1-\frac{1}{e^2}\right)\frac{14}{17}|S_{\lambda+1}(x_{\lambda+1})} 
    &\le \exp\left({-\frac{4}{49\cdot 3}\left(1 - \frac{1}{e^2}\right)\frac{14}{17}|S_{\lambda+1}(x_{\lambda+1})|}\right)
    \\\text{Lemma \ref{lemma:swarm_size} then implies,}
    \\&\le \exp\left({-\frac{4}{49\cdot 3}\left(1 - \frac{1}{e^2}\right)\frac{14}{17}\frac{c\lambda}{2}}\right)
    \\&= \exp\left({-\frac{4}{357}\left(1 - \frac{1}{e^2}\right)c\lambda}\right)
    \\\text{for $c \ge 510\cdot k$ we get,}
    \\&\le n^{-k}.
\end{align*}
\end{itemize}
\end{proof}

We conclude the analysis by observing each node's congestion. 
Therefore, we first bound the expected 
number of trajectories that cross an interval in each round.
Note that a trajectory is defined on points in $[0,1)$
and \emph{not} on actual nodes (except the first and last element).
We can see that it holds:
\begin{lemma}
\label{lemma:interval_congestion}
Assume all nodes choose their position independently and uniformly at random in the $[0,1)$-interval.
Moreover, let each node send $\varphi \in \mathbb{Z}_+$ messages to targets picked independently and uniformly at random from the $[0,1)$-interval. Then for every $I \subset [0,1)$ it holds that,
\begin{enumerate}
\item $X^j_I$ is the sum of independent $\{0,1\}$ random variable,
\item $\E{X^j_I} = kn|I|$,
\end{enumerate}
where $X_I^j$ is a random variable that counts the number of trajectories that have their $j^{\text{th}}$ step in the interval $I$ and $|I|$ denotes the size of the interval $I$. 
\end{lemma}
\begin{proof}
W.l.o.g. assume that $I := [a,b]$ with $0\le a \le b < 1$.

~

\textit{1}. Let $X^{j}_{I:(v,i)}$ be a $\{0,1\}$ random variable such that, $X^{j}_{I:(v,i)}=1$ if the trajectory of a message $i \in [\varphi]$ started by a node $v$ crosses the interval $I$ in its $j^{th}$ step and $X^{j}_{I:(v,i)}=0$, otherwise.
Then the number of messages with their $j^{th}$ step in the interval $I$ is given by,
\[ X_I^j = \sum_{v \in V} \sum_{i \leq \varphi} X^{j}_{I:(v,i)}. \]
Observe that a message's trajectory is uniquely defined 
by the starting node $v$ and the end point $p$. 
Since for each message the target node is chosen uniformly and independently at random from $[0,1)$-interval, we conclude that the set of random variables $\left(X^{j}_{I:(v,i)}\right)_{v \in V, i \in [\varphi]}$ are independent.

\textit{2}. We prove by induction that for each $j \in [\lambda+1]$ step along the trajectory, it holds that, \[\E{X^j_I} = kn|I|.\] 

~

Consider step $j=0$ i.e., the step where is messages are at their starting node at $x_0$. For each $v \in V$, let $X_v$ be a $\{0,1\}$ random variable such that, $X_v=1$ if node $v \in I$ and $X_v=0$, otherwise. Observe that the nodes pick their positions uniformly and independently at random in the $[0,1)$-interval. Therefore, \[\pr{X_v=1} = |I|.\] This implies, \[\E{\sum_{v \in V} X_v} = \sum_{v \in V}\pr{X_v=1} = n|I|.\] Since each node initiates $\varphi \in \mathbb{Z}_+$ messages to randomly picked targets in the $[0,1)$-interval, \[\E{X_I^0} = \E{\sum_{v \in V} \varphi\cdot X_v} = n\varphi|I|.\] 

Now assume the induction hypothesis holds for every $j \in [\lambda]$, this implies, \[\E{X_I^{\lambda}} = \varphi n|I|.\] Let $I_0 := I \cap [0,\nicefrac{1}{2})$ and $I_1 := I \cap [\nicefrac{1}{2},1)$ be the parts of $I$ that lie in the first and second half of $[0,1)$-interval, respectively. Then, \[ \E{X_I^{\lambda+1}} = \E{X^{\lambda+1}_{I_0}} + \E{X^{\lambda+1}_{I_1}}.\]
Observe that for each $i \in \{0,1\}$ the bit representation of each point on $I_i$ begins with $i$. Therefore for any message $m$ destined to a uniformly and independently picked target $p \in [0,1)$-interval, it's trajectory crosses interval $I_i$ in the $j^{\text{th}}$ step of the trajectory if and only if $p_j = i$, where $p_j$ is the $j^{\text{th}}$ most significant bit of $p$ in the binary representation. W.l.o.g. we only analyze $I_0$ to show that,
\[
    \E{X^{\lambda+1}_{I_0}} = \varphi n|I_0|.
\]
The proof for $\E{X^{\lambda+1}_{I_1}} = \varphi n|I_1|$ is analogous and follows using similar arguments.

~

Consider an arbitrary trajectory $\tau := x_0, \dots, x_{\lambda+1} \in [0,1)^{\lambda+2}$  with $x_{\lambda+1} \in I_0 := [a,b_0]$. From Definition~\ref{def:trajecory} it must then hold that $x_{\lambda} = 2x_{\lambda+1}$. Therefore, for all trajectories that cross the interval $I_0 \in [0,\nicefrac{1}{2})$ it holds that $x_\lambda = 2x_{\lambda+1}$ and hence, the $\lambda^\text{th}$ step must be in the interval $J := [2a, 2b_0]$. The size of the interval is then, \[|J| = 2|I_0|.\]
By the induction hypothesis we know that, \[\E{X_J^\lambda} = \varphi n|J| = 2kn|I_0|.\]

Note that since for every message it's target is picked uniformly and independently at random from the $[0,1)$-interval, this is equivalent to the thought experiment of flipping a fair coin for each bit of the target address. Therefore, the probability that a trajectory $\tau$ in interval $J$ points towards the interval $I_0$ in it's $(\lambda+1)^\text{th}$ step is then, \[\pr{\tau_{\lambda+1} \in I_0} = \frac{1}{2},\] where $\tau_{\lambda+1}$ denotes the position of the trajectory $\tau$ in step $\lambda+1$.

Then, the expected number of trajectories that point from the interval $J$ to the interval $I_0$ is given by, 

\begin{align*}
    \E{X^{\lambda+1}_{I_0}} &= \sum_{\ell=0}^\infty \E{X^{\lambda+1}_{I_0} | X^{\lambda}_{J}=\ell}\cdot \pr{X^{\lambda}_{J}=\ell} & \textit{(Law of Total Expectation)}\\
    &= \sum_{\ell=0}^\infty \sum_{\tau \in \left[1, \ldots, \ell \right]}\pr{\tau_{\lambda+1} \in I_0} \cdot\pr{X^{\lambda}_{J}=\ell} \\
    &= \sum_{\ell=0}^\infty \frac{\ell}{2}\cdot \pr{X^{\lambda}_{J}=\ell}\\
    &= \frac{1}{2} \sum_{\ell=0}^\infty \ell \cdot \pr{X^{\lambda}_{J}=\ell}\\
    &= \frac{1}{2}\E{X^{\lambda}_{J}} = kn|I_0|.
\end{align*}

Since $|I| = |I_0| + |I_1|$ we get, \[
    \E{X^{\lambda+1}_{I}} = \varphi n|I_0| + \varphi n|I_1| = \varphi n|I|.
\]
This completes the induction.
\end{proof}

~

Using Lemma~\ref{lemma:interval_congestion} we can then bound from above the expected congestion for \routing.
\begin{lemma}
\label{lemma:interval_congestion2}
If each node picks its position and the $t \in \mathbb{Z}_+$ target nodes uniformly and independently at random, then \routing has a
expected congestion at most $24rtc\lambda$.
\end{lemma}
\begin{proof}
Let $v \in V$ be any node and let $\left[v \pm \frac{c\lambda}{n}\right] =: I_v \subset [0,1)$ be an interval that contains all points $p$ with $v \in S(p)$.
Observe that a message may be routed via $v$ only if its trajectory passes the interval $I_v$.
From Lemma \ref{lemma:interval_congestion} we know that for any given round $j$ the expected number of trajectories that cross the interval $I_v$ is then, \[\E{X^{j}_{I_v}} = 2tc\lambda,\] where $X_{I_v}^j$ is a random variable that counts the number of trajectories that have their $j^{\text{th}}$ step in the interval $I_v$. Using the Chernoff bound we get,
\begin{align*}
\pr{X^{j}_{I_v} \ge 3t c\lambda} &\le \exp\left(-\frac{1}{6}tc\lambda\right)
\\\text{for $c \ge 510\cdot k$,}
\\&\le n^{-k}.
\end{align*}
Therefore, the total number of trajectories that pass interval $I_v$ in any round $j$ is at most $3tc\lambda$ w.h.p. We know from Definition~\ref{def:trajecory}, that a trajectory passing interval $I_v$ in step $j$, had it's step $j-1$ in either interval $J^0 := \left[ 2\left(v \pm \frac{c\lambda}{n}\right)\right]$ with $p_{j} = 0$, or in the interval $J^1 := \left[ 2\left(v \pm \frac{c\lambda}{n}\right)-1\right]$ with $p_{j}=1$. Observe that the size of these intervals i.e., $|J^0| = |J^1| = 2|I_v|$. From Lemma~\ref{lemma:swarm_size} we know that the number of nodes in each of these interval i.e., $J^0$ and $J^1$, are at most $8c\lambda$ w.h.p. Therefore, the total number of messages that will be forwarded to the interval $I_v$ in any given round is then at most $M = (r\cdot 8c\lambda) \cdot (3tc\lambda)$ w.h.p.

Let $(X_i^v)_{i \in [M]}$ be a set of $\{0,1\}$ random variables such that, $X_i^v=1$ if message $m_i$ is sent to node $v$ and $X_i^v=0$, otherwise. From Lemma~\ref{lemma:swarm_size} we know that any interval of size $\frac{2c\lambda}{n}$ has at least $c\lambda$ nodes w.h.p. Then, \[\pr{X_i^v = 1} \le \frac{1}{c\lambda}.\] 

The expected congestion is then,
\[\E{\sum_{i \in [M]} X_i^v} \le
\frac{1}{c\lambda}(r\cdot 8c\lambda) \cdot (3tc\lambda) = 24rtc\lambda.\]

~

\end{proof}

Theorem \ref{thm:routing} follows directly from Lemmas \ref{lemma:routing_correctness}, \ref{lemma:routing_whp} and \ref{lemma:interval_congestion2}. 
\end{proof}

\subsection{The Random Sampling Algorithm}

\begin{figure}
\begin{algorithm}[caption={Random Sampling \sampling}, label={alg:sampling}]
$\textbf{Desc:}$ This algorithm is executed on a routable graph $\mathcal{D} := (D_1, H_1, \dots)$. It routes a message $m$ from any node $u \in V_t$ to a node $v$ (almost) uniformly picked from $V_{t+2\lambda+2}$.  

$\textbf{Send a message } $m$ \textbf{ to a random node}$ $v \in V_{t+2\lambda+2}$
  $p \longleftarrow$ Uniformly chosen from $[0,1)$
  $\Delta \longleftarrow$ Uniformly chosen from $[0,2c\lambda]$
  Route message $(m,p,\Delta)$ to target $S(p)$ using $\text{\routing}$
    
$\textbf{Upon receiving}$ $(m,p,\Delta)$ $\textbf{from}$ $\text{\routing}$
  $P \longleftarrow \{w \in S(p) \mid p_w \in S(p)\}$
  Choose $w$ such that $|\{u\,| \, u \in \langle p,w \rangle\}| =  \Delta$ mod $P$ 
  Deliver $m$ to $w$

\end{algorithm}
\end{figure}
Besides routing to a random swarm $S(p)$ for some $p \in [0,1)$ the algorithm \routing can also be extended to send a message to a \emph{node} (and not swarm) chosen \emph{uniformly} at random. 
We call this algorithm \sampling.
The underlying approach is adapted from King and Saia~\cite{KingS04} and King et al.~\cite{KingLSY07}.
In their algorithm, King and Saia condition on the fact that the hash function that provides the nodes with their positions has certain properties.
In particular, these properties are fulfilled with probability at least $1-\frac{3}{n}$ on a randomly chosen hash function $h$.
This is not good enough for our case as we need a sampling routine that works correctly w.h.p., i.e., with probability $1-\frac{1}{n^k}$ where we can freely choose some $k>1$.
We strongly believe that their approach can be adapted to work, w.h.p.,
however we will only show a weaker statement that suffices for our needs.
In particular, we only need the sampling probabilities to be within a constant factor, i.e., between, $\frac{1}{4n}$ and $\frac{5}{n}$.
Our approach works as as follows. A node first picks a value $p \in [0,1)$ uniformly at random and routes the message to the swarm $S(p)$ using algorithm \routing.
Then, the message is \emph{only} delivered to some randomly chosen node $w \in S(p)$.
For this, the message includes a random chosen number $\Delta \in [0,2c\lambda]$.
The message is then delivered to the node for which it holds $|\{u\,| \, u \in \langle p,w \rangle\}| = \Delta \textit{ mod } |S(p)|$, i.e., the $\Delta^{th}$ node (in $S(p)$) that is right of $p$.
Since all nodes in $S(p)$ know $|S(p)|$ (because they know the IDs of all nodes in $S(p)$) and $\Delta$ (because they all received the message) this can checked locally without further messages.
The following lemma bounds the sampling probability.
\begin{lemma}
\label{lemma:sampling}
Let $\mathcal{D}$ be routable. Assume, a node $v \in V_t$ starts a message $m_v$ using \sampling.
Then for all $u \in V_{t + 2\lambda+2}$, 
\[ \Pr[u \textit{ receives } m_v] \in \left[\frac{1}{4n},\frac{5}{n}\right].\]
Furthermore, for any two nodes $v,w \in V_t$ that start messages $m_v$ and $m_w$, respectively, it holds that for all $u \in V_{t + 2\lambda+2}$,
\[ \Pr[u \textit{ receives } m_v] = \Pr[u \textit{ receives } m_w]. \]
\end{lemma}
\begin{proof}
Let $Y(v,u)$ be a $\{0,1\}$ random variable such that $Y(v,u)=1$ in the event that $v$ samples $u$ i.e., the message $m_v$ is delivered to $u$.
Further, let $p \in [0,1)$ be the point that $v$ chooses in line $4$ of \sampling. 
Then, we make the following observations.
\begin{enumerate}
    \item As a necessary condition to sample the node $u$, node $v$ must pick some point $p \in S(u)$. Otherwise, $u$ will never be considered in the second step.
    This happens with probability $\frac{2c\lambda}{n}$.
    \item Given that $p \in S(u)$, \sampling still needs to pick $u$ uniformly from all nodes in $S(p)$. The probability for this depends only on $|S(p)|$ and $\Delta$.
    By Lemma \ref{lemma:swarm_size} we know that w.h.p. it contains at most $2c\lambda$ nodes and at least $c\lambda/2$ nodes. 
    Since $\Delta \leq 2c\lambda$ and $\frac{c}{2}\lambda \leq |S_l(p)| \leq 2c\lambda$ there are at least one and at most $4$ choices of $\Delta$ that result in $u$ being picked.
\end{enumerate}
These observations are sufficient to prove the statement. 
For the lower bound we get that:
\begin{align*}
    \pr{Y(v,u)} &= \pr{p \in S^L(u)} \cdot \pr{\text{u is picked from }S^R(p)}\\
    &\geq \pr{p \in S^L(u)} \cdot \left(\pr{|S^R(p)| \leq 2c\lambda}\frac{1}{2c\lambda}+\pr{|S^R(p)| > 2c\lambda}\frac{1}{n}\right)\\
    &= \frac{c\lambda}{n} \cdot \left( \left(1-\frac{1}{n^k}\right)\frac{1}{2c\lambda}+\frac{1}{n^{k+1}}\right)\\
    &\ge \frac{1}{2n}  - \frac{1}{2n^{k+1}}\\
    &\geq \frac{1}{4n}.
\end{align*}
The proof for the upper bound is analogous, we simply need to replace the upper bound for the swarm size with the lower bound.
We get
\begin{align*}
    \pr{Y(v,u)} &= \pr{p \in S^L(u)} \cdot \pr{\text{u is picked from }S^R(p)}\\
    &\leq \pr{p \in S^L(u)} \cdot \left(\pr{|S^R(p)| \geq \frac{c}{2}\lambda}\frac{1}{c\lambda}+\pr{|S^R(p)| \leq \frac{c}{2}\lambda}\cdot 1\right)\\
    &= \frac{2c\lambda}{n} \cdot \left(\left(1-\frac{1}{n^k}\right)\frac{1}{c\lambda}+\frac{1}{n^{k}}\right)\\
    &\le \frac{c\lambda}{n} \cdot \left(\frac{4}{c\lambda}+\frac{1}{n^{c}}\right)\\
    &= \frac{4}{n} + \frac{2c\lambda}{n^{k+1}}\\
    &\leq \frac{5}{n}.
\end{align*}
This concludes the first part of the lemma.

It remains to show that (messages of) nodes $v$ and $w$ are delivered to a node $u$ with the same probability.
Let $p_v$ and $p_{w} \in [0,1)$ be the points picked by these nodes for their respective messages. 
Further, let $\Delta_v$ and $\Delta_w \in [0,2c\lambda]$ be random numbers.
Let $\ell(u,p) :=  |\{w\,| \, w \in \langle p,u \rangle\}|$ and $\mathcal{U}(u,p)$ be the number of possible choices of $\Delta$ that lead to $u$ being picked given that the message is routed to $p$ in the first step.
Note that any point $p \in [0,1)$ and any $\Delta \in [0,2c\lambda]$ is picked with equal probability by $v$ and $w$.
Therefore, it holds that,
\begin{align*}
\pr{Y(v,u)} &= \pr{p_v \in S^L(u)} \cdot \pr{\text{u is picked from }S^R(p_v)}\\
    &= \pr{p_v \in S^L(u)} \cdot \left(\sum_{p^{*} \in S^L(u)} \pr{p_v=p^*}\cdot\pr{\ell(u,p^*) = \Delta_v \textit{ mod }|S^R(p^*)|}\right)\\
    &= \pr{p_v \in S^L(u)} \cdot \left(\sum_{p^{*} \in S^L(u)} \pr{p_v=p^*}\frac{\mathcal{U}(p^*)}{|S^R(p^*)|}\right)\\
    &= \pr{p_{w} \in S^L(u)} \cdot \left(\sum_{p^{*} \in S^L(u)} \pr{p_{w}=p^*}\frac{\mathcal{U}(p^*)}{|S^R(p^*)|}\right)\\
    &= \pr{p_w \in S^L(u)} \cdot \left(\sum_{p^{*} \in S^L(u)} \pr{p_w=p^*}\cdot\pr{\ell(u,p^
    *) = \Delta_w \textit{ mod }|S^R(p^*)|}\right)\\
    &= \pr{p_{w} \in S^L(u)} \cdot \pr{\text{u is picked from }S^R(p_w)}\\
    &= \pr{Y(w,u)}.
\end{align*}

The third equality is due to the fact that, \[\pr{p_v \in S^L(u)} = \pr{p_w \in S^L(u)} = \frac{c\lambda}{n},\] and for any $p \in 
[0,1),$\[\pr{p_v = p} = \pr{p_w = p}.\]
This follows from the fact that all nodes use the same random hash function $h$ to compute the positions.
This concludes the proof.
\end{proof}

\section{The Maintenance Algorithm}
\label{sec:maintainer}

In this section we present our main contribution.
Our algorithm maintains a \emph{routable} dynamic overlay $\mathcal{D} = (D_0, H_0, D_1, \dots)$ with high probability.
Before we present the algorithm, we will first give an overview on our assumptions and choice of churn parameters.

We assume a $(2,2\lambda+7)$-late adversary with a churn rate  $(\nicefrac{n}{16}, 2\lambda+7)$.
This implies that within $2\lambda+7$ (i.e. $O(\log n)$) rounds, a constant fraction of nodes can be subjected to churn.
Further, we assume that the number of nodes in any round is at most $(1+\nicefrac{1}{16})n$.
Note that the values $\alpha = \nicefrac{1}{16}$ and $\kappa = (1 + \nicefrac{1}{16})$ are chosen for the sake of convenience in analysis.
We require a bootstrap phase of length $B := 2\lambda+2$ at the beginning of the algorithm.
In this phase no churn occurs and this enables us to initialize our algorithm. 
Such a bootstrap phase is a standard assumption in the area of churn-resistant overlays and is very likely necessary to construct a robust overlay.

We also assume that the system starts in an initial LDS $D_0$ in round $0$.
This assumption is made for convenience as the initial overlay can easily be constructed 
in the churn-free bootstrap phase using algorithms from \cite{GmyrHSS17,GoetteHSW20}. 
Using their techniques this can be achieved in $O(\log^2 n)$
rounds with a deterministic algorithm or in $O(\log n)$ rounds w.h.p. with a randomized algorithm. 
Both these algorithms assume that the congestion and degree of each node is polylogarithmic, so they fit well into our computational model.
We would like to remark that since our focus lies on fast reconfiguration and not on optimizing the bootstrap phase we omit the algorithmic details and the corresponding analysis. 
For ease of notation we will refer to round $t+B$ simply as $t$. 

Let $\overline{V_t} := V_t \cap V_{t-1}$ denote the set of all rounds except for the newly joined nodes in any given round $t$.
Over the course of this section we distinguish between three types of nodes in each round $t$. 
Namely, the set of mature nodes $M_t \subseteq \overline{V}_t$, which are nodes that are in the network for at least $2\lambda+2$ rounds (or $2\lambda+3$ if they joined in an odd round), 
the set of fresh nodes $F_t := \overline{V}_t\setminus M_t$ which are nodes that are at least one round, but less than $2\lambda +2$ rounds old, and the set of newly joined nodes i.e., $V_t\setminus \overline{V}_t$. 
Observe that $\overline{V}_t := M_t \cup F_t$ and due to our choice of churn parameters, it holds $M_t \geq n(1-\frac{1}{16})$ and $F_t \leq \nicefrac{n}{16}$. 

Our algorithm is comprised of two subroutines, \maintainer and \random, that are executed in concurrently.
\maintainer ensures that all mature nodes build a routable overlay each round and \random makes sure that all fresh nodes (which are not part of the routable overlay) stay connected to $O(\log n)$ mature nodes until they mature themselves. This ensures that the matures node can route a message on behalf of the fresh nodes over the overlay.
The main result of this section (and this paper) is stated in the following theorem.
\begin{theorem}
Subroutines \maintainer and \random maintain a series of overlays $\mathcal{D}$ such that for $O(n^k)$ rounds w.h.p,
\begin{enumerate}
\item the mature nodes form a routable series of graphs $\mathcal{D} := (D_0,H_0,D_1, \dots)$, 
\item each fresh node is known by $\Theta(\log n)$ mature nodes, and
\item the congestion is $O(\log^3 n)$ per node and round.
\end{enumerate}
\end{theorem}

We would like to remark that both \maintainer and \random are heavily randomized and can possibly fail to create a connected and routable overlay if they are executed for too long. 
For example, the algorithm could fail if certain swarms are too small and/or too many messages are dropped by the routing algorithm.
In these cases, the algorithm cannot construct the desired overlays and needs to be restarted along with another bootstrap phase.
Given that the algorithm runs correctly \emph{w.h.p.}, i.e. no failure happens with probability $O(n^{-k})$ for a tunable constant $k>0$, we can only guarantee that the algorithm runs smoothly for $O(n^k)$ rounds w.h.p. until some failure happens.
This follows by a simple application of the union bound.
Throughout this chapter we assume that $\frac{k\log(n)}{n} << 1$, i.e., both $k$ and $\log(n)$ are very small compared to $n$. Therefore, our algorithms only become applicable for large values of $n$ (say $n > 10^6k$).
However, as our goal is to show that the messages per node stays logarithmic in $n$ even under heavy churn, we believe that it is justified to only consider very high values of $n$. 
In particular, we do not seek (nor claim) that we make the optimal choice of constants.

\maintainer and \random exchange four types of messages between the nodes to build a series of overlay.
\begin{enumerate}
    \item The message $\textsc{Connect}(v)$ is sent by a fresh node $v$ to advertise itself to a mature node in the overlay. It only contains $v$'s identifier.
    \item The message $\textsc{Create}(v,p^t_v)$ is used to introduce a node $v$ to its neighbours in $S_t(p^t_v)$.
    The message contains the node's ID and its position $p^t_v$ in the overlay $D_t \in [0,1)$.
    \item The message $\textsc{Join}(v,p^t_v)$ is used to introduce a node $v$ to nodes in $S_{t-1}(p^t_v)$.
    It is routed from its origin to position $p^t_v$ in overlay $D_t$. It contains the identifier of $v$ and its position $p_v^{t} \in [0,1)$.
    \item The message $\textsc{Token}(v)$ is sent by a mature node $v$ to a point in the $[0,1)$-interval picked (almost) uniformly at random. It only contains $v$'s identifier.
\end{enumerate}

\subsubsection*{Building a Routable Overlay}
After the bootstrap phase, the algorithm \maintainer creates a series of overlays $\mathcal{D}=\left(D_0, H_0, D_1, H_1, \dots \right)$ that contain all \emph{mature} nodes in any given round. 
In particular, in every \emph{even} round $t = 2i$ the algorithm creates an LDS $D_i$ which consists of all mature nodes $\overline{V}_{t}$. 
In each \emph{odd} round $t = 2i+1$ the algorithm creates a handover graph $H_i$ in which for each $p \in [0,1)$ it holds that $S_{i}(p)$ and $S_{i+1}(p)$ are adjacent (Lemma~\ref{lemma:swarm_propery}), where $S_{i}(p)$ and $S_{i+1}(p)$ are the swarms of point $p$ in $D_i$ and $D_{i+1}$, respectively.

To construct a series of overlays, all mature nodes continuously choose new positions $p^0_v, p^1_v, \dots$ for the corresponding overlays $D_0, D_1, \dots$ and use the routing algorithm to find their neighbors.
More precisely, the construction of $D_i$ begins in round $2i-(2\lambda+2)$. Every node $v \in \overline{V}_{2(i-\lambda-1)}$
picks a position $p^{i}_v \in [0,1]$ uniformly at random and routes its $ID$ to the target $p^{i}_v$ along the trajectory $(v, x_1, \cdots, x_\lambda, p^{i}_v)$.
This message arrives $2\lambda$ rounds later, i.e., in round $2i-2$ at the swarm $S_{i-1}(x_\lambda)$.
Then, within two rounds, \maintainer first constructs the handover graph $H_{i-1}$ and a new LDS $D_i$ based on these positions. 
This ensures that in the even rounds the forwarding step from \routing can be performed and the handover in the odd rounds. 
Thus, \maintainer maintains a routable overlay.
\begin{figure}
\begin{algorithm}[caption={Overlay Maintenance Algorithm \maintainer}, label={alg:maintainer}]
$\textbf{Desc:}$ In every even round $2t$ the algorithm creates a LDS $D_t$ consisting of all nodes that joined the network before round $2t - (\lambda+2)$ and performs the Forwarding Step from $\text{\routing}$. In every odd round $2t+1$ the algorithm performs a handover from $D_t$ to $D_{t+1}$ using a helper graph $H_t$.

$\textbf{Note:}$ The following code is executed by each node $u \in M_{t}$ every even and odd $\text{round}$ respectively. The messages are handled in the given order. The last block of commands in each phase is executed after all messages have been handled. 

$\rule[1pt]{6.2cm}{1pt}$ $\textbf{Even Round}$ $\rule[1pt]{6.2cm}{1pt}$
$\textbf{Upon receiving}$ $CREATE(v,p^t_v)$ $\textbf{from}$ $u'$
  $D^u_t \longleftarrow D^u_t \cup \{(v,p^t_v)\}$ $\hspace*{2.3cm} \vartriangleright$ $u$ creates edges to these nodes

$\textbf{Upon receiving}$ $JOIN(v,p^{t+1}_v)$ $\textbf{from}$ $\text{\routing}$
  Send $JOIN(v,p^{t+1}_v)$ to all nodes $w \in D^u_t$ with $p^t_w \in \langle p_v^{t+1} \pm \frac{2c\lambda}{n} \rangle \cup \langle \frac{p_v^{t+1}}{2} \pm \frac{3c}{2n}\lambda \rangle \cup \langle \frac{p_v^t+1}{2} \pm \frac{3c}{2n}\lambda \rangle$

$\textbf{Finally}$
  Perform $\emph{Forwarding Step}$ from $\text{\routing}$ using edges created from $D_{t}$
  $C \longleftarrow $ All fresh nodes known by $u$ (provided trough $\text{\random}$)
  $\forall v \in C \cup \{u\}$ do:
    $p_v \longleftarrow$ $h(v,t)$
    Route message $JOIN(v,p^{t+\lambda+1}_v)$ to target $p^{t+\lambda+1}_v$ using Algorithm $\text{\routing}$
    
$\rule[1pt]{6.2cm}{1pt}$ $\textbf{Odd Round}$ $\rule[1pt]{6.2cm}{1pt}$
$\textbf{Upon receiving}$ $JOIN(v,p_v)$ from a node $u'$
  $H_{t} \longleftarrow H_t \cup \{(v,p_v)\}$ $\hspace*{2.3cm} \vartriangleright$ $u$ creates edges to these nodes
        	
$\textbf{Finally}$
  Perform $\emph{Handover Step}$ from $\text{\routing}$ using edges created from $H_{t}$
  $\forall (v,p^{t+1}_v) \in H_{t+1}$ do:
    Send $CREATE(v,p^{t+1}_v)$ to all nodes $(w,p^{t+1}_w) \in H_{t}$ with $p^{t+1}_w \in \langle p_v^{t+1} \pm \frac{2c\lambda}{n} \rangle \cup \langle \frac{p_v^{t+1}}{2} \pm \frac{3c}{2n}\lambda \rangle \cup \langle \frac{p_v^t+1}{2} \pm \frac{3c}{2n}\lambda \rangle$
   
\end{algorithm}
\end{figure}

The algorithm proceeds in rounds. In any given round, every mature node picks a random position $p \in [0,1)$ and routes its $ID$ using message $\textsc{Join}(ID, p)$ to the respective target points using \routing. 
This is the position the node will occupy in $D_i$ in round $2i$.

It uses the routing algorithm \routing to send its ID to the swarms $S_{i-1}(p)$, $S_{i-1}(\nicefrac{p}{2})$, and $S_{i-1}(\nicefrac{p+1}{2})$ thereby creating the handover graph $H_{i-1}$ in round $2i-1$. 
Starting from the handover graph $H_{i-1}$, the overlay $D_i$ can then be created in a single additional round through local introductions.

We will now describe the construction of the overlays $H_{i-1}$ and $D_{i}$ in detail given that the algorithm worked correctly until that round.
We assume the system is currently in an \emph{even} round $t=2i-2$ and in all previous rounds the mature nodes formed a routable overlay $\mathcal{D} := D_0, H_0, \dots, D_{i-1}$.
In other words, all messages of the from $(v,p^i_v)$ where $v$ is an identifier and $p^i_v$ is its position in $D_i$, are one round away from reaching their target.
Then, the construction of $H_{i-1}$ and $D_i$ is as follows.

\begin{enumerate}
    \item In round $2i-2$ \routing executes the forwarding step. This implies, all messages are routed to their target location.
    In particular, each message with target point $p$ will be received by all nodes in $S_{i-1}(p)$.
    Recall that this is ensured by the fact that the message is sent to all nodes in $S_{i-1}(p)$ in the last step of the trajectory. Additionally, note that \maintainer ensures that the message is also forwarded to all nodes in $\langle p \pm \frac{2c\lambda}{n}\rangle$, $\langle\frac{p}{2} \pm \frac{3c\lambda}{2n}\rangle$, and $\langle\frac{p+1}{2} \pm \frac{3c\lambda}{2n}\rangle$. 
  
    \item 
    In round $2i-1$, each node $w \in \langle p^i_v \pm \frac{2c\lambda}{n}\rangle \cup \langle\frac{p^i_v}{2} \pm \frac{3c\lambda}{2n}\rangle \cup \langle\frac{p^i_v+1}{2} \pm \frac{3c\lambda}{2n}\rangle$ receives messages of the form $\textsc{Join}(v,p^{i}_v)$.
    Therefore, the construction of the handover graph $H_{i-1}$ follows directly from the correctness of \routing and Definition~\ref{def:handover_graph}.
    In the remainder of round $2i-1$ two routines are processed concurrently. 
    First, the handover step of \routing is executed. 
    Second, all nodes must learn their neighbors in $D_i$ in order to execute the forwarding step in round $2i$. 
    For this, the nodes iterate over all received messages of the form $(v,p^{i}_v)$ and \emph{introduces} them to all their neighbors.
    By introduction, we mean that the neighbor's $ID$ and position is sent to $v$.
    As we will see, for every pair of neighbors in $D_i$ there is at least one node that knows both their IDs  and introduces them.
    These messages arrive in round $2i$.
    \item Finally, at the beginning of round $2i$, each node knows all its neighbors in new overlay $D_{i}$ (and all messages that need to be forwarded in $D_i$). 
    Thus, all mature nodes form the overlay $D_{i}$ (and can perform the forwarding step). 
\end{enumerate}
Note that after round $2i+1$ the nodes' positions in $D_i$ and $D_{i+1}$ are in no relation with each other and hence, the edges in $D_{i+1}$ are independent of $D_i$.
Therefore, the adversary stays oblivious of all nodes' current positions.

Further, observe that our approach requires that \emph{both} the fresh and the mature nodes send out the join requests and that all messages take \emph{exactly} the same time to reach its destination.
The latter is ensured through \routing.
For the former we assume that each fresh node is known by least one mature node, which is part of $D_i$. 
However, this will be ensured by \random and explained in the next section. 

Listing \ref{alg:maintainer} presents the pseudocode for the algorithm.
Each node has the two variables $D^u_t$ and $H^u_t$. $D^u_t$ stores $u$'s neighborhood in $D_t$ whereas $H^u_t$ stores the references for the handover. Both variables may be reset at the end of each round. 
The nodes pick a random position in the $[0,1)$-interval using a uniform hash function $h: V \times \mathbb{N} \to [0,1)$ known to all nodes. 
This hash function can either be established in the bootstrap phase by sending $O(\log^2 n)$ bits of shared randomness to each node or we can assume it is known to all nodes. 
The former case, it would have to be renewned every $O(\log (n))$ as the adversary could gain access to it.
One way to do this would be to transform the routing algorithm into a broadcast algorithm that can send a given message to all nodes. For brevity we omit this and simply assume that all nodes know $h$ and the adversary does not have access to it.

This hash function takes the node's $ID$ and the current round as an input and computes a random value $p_v$. 
Note that this choice excludes some points in $[0,1)$ from being picked as all values need to be encoded in $O(\log n)$ bits.
However, this does not impact the correctness of our algorithms.
Instead, we just handle the values returned by $h$ as continuous values as it is a standard assumption (see, e.g. \cite{KingS04, KingLSY07} for the usage in overlay networks and \cite{bellare1993random} for a proof that these functions can indeed be simulated by few random bits).

Recall that the fresh nodes are not part of the overlay. Therefore, the mature nodes send out requests on behalf of each fresh node $u \in F_t$ known to them. Note that each node can compute $h(u,t)$ if it knows $u$'s ID. The $ID$s of these nodes are stored in the variable $C$.
This variable is set by \random. 
Details on how it set can be found in the next section.

\subsection{Analysis of \maintainer}
\label{sec:maintainer_analysis}
In this section we show that \maintainer maintains a dynamic overlay
with the properties needed for routing. 
Throughout this section we assume that \random works correctly and each fresh 
node is connected to $\Theta(\log n)$ mature nodes at any time. 
Thus, \emph{every} node in the networks starts a join request in every even round.
\begin{lemma}
\label{lemma:correctness_lds}
Let $\mathcal{D}_t := (D_0, H_0 \dots, D_i)$ be routable graph until round $t=2i$.
Then it holds $\mathcal{D}_{t+2} := (D_0, H_0 \dots, D_i,H_i,D_{i+1})$ is a routable graph until round $t+2$ w.h.p.
\end{lemma}
\begin{proof}
W.lo.g. we assume that the algorithm is currently in round $t = 2i$ 
and the overlays $\mathcal{D}_t := (D_0, H_0 \dots, D_i)$ were routable.
This implies that the mature nodes know all neighbors in $D_i$, all join requests $(v,p_v^{i+1})$ started $2\lambda$ rounds ago are delivered, and the nodes are ready to perform the final forwarding step of the messages. Due to Lemma~\ref{lemma:good_nodes} we know that $D_i$ is good, and at least $\nicefrac{3}{4}$-fraction of each swarm in $D_i$ survives until $2i+1$. We will now show that \maintainer maintains the following three properties.
\begin{enumerate}
\item \maintainer successfully constructs $H_i$ in round $2i+1$,
\item constructs a new LDS $D_{i+1}$ in round $2i+2$, and
\item all swarms $S_{t+2}(p)$ are good \emph{w.h.p}.
\end{enumerate}
Together, these three properties imply that $\mathcal{D}_{t+2} := (D_0, H_0 \dots, D_i,H_i,D_{i+1})$ are a series of routable overlay.

~

The following lemma shows that \maintainer constructs $H_{i}$ in round $2i+1$, i.e., we show that for any $p \in [0,1)$, every node in $S_{2i}(p)$ knows the ID of every node in $S_{2i+1}(p)$. The proof essentially follows using correctness of \routing and Lemma~\ref{lemma:swarm_propery}.
\begin{lemma}[Correctness of the Handover Construction]
\label{lemma:spread}
Let $\mathcal{D}_t := (D_1, H_1 \dots, D_i)$ be routable graph until round $t=2i$.
Then, in round $2i+1$, each node in $\langle p^{i+1}_v \pm \frac{2c\lambda}{n}\rangle \cup \langle\frac{p^{i+1}_v}{2} \pm \frac{3c\lambda}{2n}\rangle \cup \langle\frac{p^{i+1}_v+1}{2} \pm \frac{3c\lambda}{2n}\rangle$ receives $(v,p^{i+1}_v)$ w.h.p.
This implies that in round $2i+1$ the nodes form the Handover graph $H_{i}$.
\end{lemma}
\begin{proof}
The proof follows directly using correctness of \routing and the overlay's topology. 
Since $\mathcal{D}_t$ is routable until (and including) round $2i$, all messages that were started in round $2i-2\lambda$ are correctly routed to their target swarm w.h.p. via \routing.
This includes all \textsc{Join}($v,p^{i+1}$) messages that were started in round $2i-2\lambda$. Due to Lemma~\ref{lemma:good_nodes} we know that $D_t$ is good w.h.p.
Thus, in round $2i+1$ \emph{every} node in $\langle p^{i+1}_v \pm \frac{2c\lambda}{n}\rangle \cup \langle\frac{p^{i+1}_v}{2} \pm \frac{3c\lambda}{2n}\rangle \cup \langle\frac{p^{i+1}_v+1}{2} \pm \frac{3c\lambda}{2n}\rangle$ received \textsc{Join}($v,p^{i+1}$) and therefore knows $v$ w.h.p. The proof then follows from the definition of the Handover graph $H_i$.
\end{proof}


~

We continue with the construction of $D_{i+1}$. 
In particular, we show that every \emph{mature} node $v \in S_{i+1}(p^{i+1}_v)$ creates an edge to each of its new neighbors in $v' \in S_{i+1}(p^{i+1}_v) \cup S_{i+1}(\frac{p^{i+1}_v}{2}) \cup S_{i+1}(\frac{p^{i+1}_v+1}{2})$. We divide the neighbors into two sets.
\begin{enumerate}
\item the list neighbors left and right of $p^{i+1}_v$, and
\item the DeBruijn neighbors left and right of $\frac{p^{i+1}_v}{2}$ and $\frac{p^{i+1}_v+1}{2}$. 
\end{enumerate}
The following lemma show that for all nodes $v$ and $v'$ which will be neighbors in $D_{i+1}$, \emph{w.h.p.} there is at least one node $w$ that receives the messages $(v,p^{t+1}_v)$ and $(v',p^{t+1}_{v'})$ in round $2i+1$ and thus introduces the nodes.
\begin{lemma}
\label{lemma:overlap}
Let $v,w$ be any two neighbors in $D_{i+1}$, then w.h.p. 
\[|\{ u \in G_{2i+1}| u \textit{ receives } (v, p^{i+1}_v) \textit{ and } (w,p^{i+1}_{w}) \}| \geq \frac{14}{17}c\lambda,\]
where $G_{2i+1} \subseteq V_{2i+1}$ are the set of good nodes in the round $2i+1$.
\end{lemma}
\begin{proof}
Consider two nodes $v$ and $w$ with $d(p^{i+1}_v,p^{i+1}_w) \leq \frac{2c\lambda}{n}$, i.e., neighbors in $D_{i+1}$.
W.l.o.g. we assume that $p^{i+1}_w$ is right of $p^{i+1}_v$ and $d(p^{i+1}_v, p^{i+1}_w) = \frac{2c\lambda}{n}$. We make these simplifying assumptions since (a) the proof is analogous for the left and right side and (b) any closer point can only have more nodes for the introduction.

Observe that the last step of \routing is executed in round $2i$. Particularly, \maintainer ensures that the message $\join(v,p_v^{i+1})$ and $\join(w,p_w^{i+1})$ is forwards to every node in the interval $\left[ p^{i+1}_v, p^{i+1}_v + \frac{2c\lambda}{n}\right]$ and $\left[ p^{i+1}_w, p^{i+1}_w - \frac{2c\lambda}{n}\right]$, respectively. This is possible due to the topology of $D_i$. This implies there is a interval $I \in [0,1)$ of length $\frac{2c\lambda}{n}$ such that all nodes belonging to this interval receive both $\join(v,p_v^{i+1})$ and $\join(w,p_w^{i+1})$ in round $2i+1$. The claim then follows using Lemma~\ref{lemma:swarm_size} that interval $I$ has at least $c\lambda$ nodes w.h.p. and Lemma~\ref{lemma:good_nodes} that at least $\frac{14}{17}$ of those nodes in the interval $I$ are good nodes and remain until round $2i+2$ w.h.p.
\end{proof}

The next lemma shows that a $(2, 2\lambda+7)$-late adversary effectively reduces the adversarial churn to a randomized churn as the adversary is oblivious to which nodes belongs to which swarm in any given round.
\begin{lemma}\label{lemma:oblivious_lds}
A $(2,2\lambda + 5)$-late adversary enables \maintainer construct $D_{i+1}$ independent of $D_i$. 
\end{lemma}
\begin{proof}
The proof follows from the correctness of \routing and \maintainer. Recall that in every even round, each mature node in $v \in M_{2i-(2\lambda+2)}$ picks a position $p_v^{i}$ in $D_{i}$ for itself and also for the fresh nodes that are connected to them. This is the position the node $v$ occupies in round $2i$ i.e., LDS $D_i$. $\join(v, p_v^i)$ is routed using \routing and arrives at $p_v^i$ in round $2i-1$ i.e., $H_{i-1}$ and then within a round \maintainer constructs $D_i$. Therefore, a $(2, 2\lambda+5)$-late adversary is oblivious to the position of node $v$ until round $2i+2$. However, observe that \maintainer ensures that the node is at the position $p_v^{i+1}$ picked uniformly at random from $[0,1)$-interval in round $2i+2$. This implies that the position of node $v$ in $D_{i+1}$ is independent of $D_i$. 
\end{proof}


~

Finally, Lemma \ref{lemma:correctness_lds} follows from Lemma~\ref{lemma:spread},~\ref{lemma:overlap}, and~\ref{lemma:oblivious_lds}. This concludes the analysis of the maintenance algorithm.
\end{proof}
\subsubsection*{Handling New and Fresh Nodes}
\begin{figure}
\begin{algorithm}[caption={Random Overlay Algorithm \random}, label={alg:random}]
$\textbf{Desc:}$ In each round $t$ each fresh node connects to $\delta$ mature nodes that joined at least $t-(2\lambda+5)$ rounds ago.

$\textbf{Note:}$ The following code is executed by each node $u \in V$ every round $t$. All types of messages are received in the given order. The last block of commands is executed after all messages have been handled.

$\rule[1pt]{6cm}{1pt}$ $\emph{Round }t$ $\rule[1pt]{6cm}{1pt}$
$\textbf{Upon receiving}$ $TOKEN(v)$ from node $u'$:
  $T \longleftarrow T \cup \{v\}$ $\hspace*{3.25cm} \vartriangleright$ Tokens ready to be used

$\textbf{Upon receiving}$ $CONNECT(v)$ from node $v$:
  if $\exists i \in [0,{2\delta}]$ with $c_i = \bot$ 
    $i \longleftarrow$ number chosen uniformly from all $i \in [0,{2\delta}]$ with $c_i = \bot$
    $c_i \longleftarrow v$
    
$\textbf{Upon}$ $v$ $\textbf{joining}$:
  $(w_1, \dots, w_\delta) \longleftarrow$ $\delta$ tokens chosen u.i.r from $T$
  Send $CONNECT(v)$ to all $w_1, \dots, w_\delta$ $\hspace*{1cm} \vartriangleright$ $u$ sends on behalf of $v$
  $(w_1, \dots, w_\delta) \longleftarrow$ $\delta$ tokens chosen u.i.r from $T$
  Send $TOKEN(w_1), \dots, TOKEN(w_\delta)$ to $v$ $\hspace*{.77cm} \vartriangleright$ Supply $v$ with tokens
  
$\textbf{Upon receiving}$ $TOKEN(v)$ through Algorithm $\text{\sampling}$:
  $x \longleftarrow $ uniformly chosen from $\{0,1\}$
  if $x = 0$:
    $c_i \longleftarrow$ random element from $c_i, \dots, c_{2\delta}$ 
    Send $TOKEN(v)$ to $c_i$ (or discard if $c_i = \bot$)
  else if $u$:
    $T \longleftarrow T \cup \{v\}$ $\hspace*{3.05cm} \vartriangleright$ Tokens ready to be used

$\textbf{Finally}$:
  if $u$ is fresh:
    $(w_1, \dots, w_\delta) \longleftarrow$ $\delta$ tokens chosen u.i.r from $T$
    Send $CONNECT(u)$ to all $w_1, \dots, w_\delta$ 
  else if $u$ is mature:
    Send $TOKEN(v)$ to $\tau$ random nodes using Algorithm $\text{\sampling}$
    
  $(c_1, \dots, c_{2\delta}) \longleftarrow (\bot, \dots, \bot)$ $\hspace*{3.52cm} \vartriangleright$ Reset $ID$s
  $T \longleftarrow \emptyset$ $\hspace*{4.5cm} \vartriangleright$ Drop unused tokens
\end{algorithm}
\vspace*{-0.8cm}
\end{figure}
We now present \random in detail. This algorithm ensures that each fresh node is known by $\delta \in \mathcal{O}(\log n)$ randomly chosen \textit{good} mature nodes each round \emph{w.h.p.}
Algorithm is executed in rounds on nodes set $F_t$ and $M_t$ corresponding to the fresh and mature nodes of round $t$, respectively. Recall that every fresh node joins the network via a node which has been in the network for at least two rounds. This enables the bootstrapping node to update the newly joined node with IDs of $O(\log (n))$ mature nodes and also advertise the ID of the newly joined node to $O(\log (n))$ mature nodes in the overlay in the subsequent round.
Each fresh node $f \in F_t$ which is at least one round old, advertises its own ID to $O(\log (n))$ mature nodes in the overlay.
Every such unique advertisement a mature node receives, is associated with a unique key in $[0, O(\log (n))]$ and stored in its memory.
Each mature node $m \in M_t$, uniformly and independently at random samples $O(\log (n))$ other mature nodes in the overlay using ALG-SAMPLING. Each sampled ID of a matured node is either sent to a newly joined node (i.e. less than a round old and bootstrapped via $m$) with probability $p = 1/2$ or is with probability $1-p$ forwarded to the ID of a fresh node, if available, whose key is picked uniformly at random from $[0, O(\log(n))]$.

Note that at the end of round $t$ a node forgets all its incoming connections from fresh nodes and the assignment of numbers to $ID$s is reset.
Last, note that the bootstrap phase ends once the first tokens reach their target.

Listing \ref{alg:random} depicts the pseudocode for \random. We use two types of messages, $TOKEN(v)$ and $CONNECT(v)$. Both messages only contain a nodes $v$'s $ID$. The former is
used to spread the mature nodes' $ID$s, the latter is used send a fresh node's $ID$ for sampling. Note that all token that are ready to create an are stored in the variable $T$. Further, the array $(c_1, \dots, c_{2\delta})$ stores the assignment of numbers to $ID$s. It holds $c_i = v$ if $v$'s $ID$ is assigned to $i$. If no $ID$ is assigned to $i$ we write $c_i = \bot$. Note that the set $C$ mentioned in Listing \ref{alg:maintainer} consists of all $c_i \neq \bot$.
Last, note that a node can distinguish whether it received a $TOKEN(v)$ message through Algorithm \sampling, i.e., in step $1$ of the sampling process sketched above, or directly from a node, i.e., in step $3$. 

\subsection{Analysis of \random}
In this section we show that every fresh node is able is send its $ID$ to $\delta$ mature nodes each round \emph{w.h.p.} and thus stays connected to the network. In particular, we assume that $\delta \in O(\log n)$.
Therefore, we prove the following lemma.
\begin{lemma}[Random Overlay Lemma]
\label{lemma:random-overlay}
Assume that until round $t-1$ each fresh node was connected to at least $\frac{1}{2}\delta$ good nodes each round. Then, it holds w.h.p. that each $v \in F_t$ successfully connects to $\frac{1}{2}\delta$ good nodes.
\end{lemma}
We prove the lemma in several steps. 
First, we show that each node receives $\Omega(\tau)$ tokens $\emph{w.h.p}$.
To prove this we make use of a simple balls-into-bins argument.
Recall that each mature node in starts $\tau$ tokens in round $t-(2\lambda+5)$ that reach
their random destination in round $t$. 
Further, we can show that the tokens are uniformly distributed among all nodes.
\begin{lemma}
\label{lemma:equal2}
Assume Lemma \ref{lemma:random-overlay} held until round $t-1$. 
Further, let $X(\theta,v)$ denote the event that any token $\theta$ reaches $v$ in round $t$.
Then the following statements hold:
\begin{enumerate}
\item Any token reaches $v \in V_t$ with the same probability, i.e., 
\[
    \pr{X(\theta,v) = 1} = \pr{X(\theta',v) = 1}. 
\]
\item For each token $\theta$, it holds
\[
    \pr{X(\theta,v) = 1} \geq \frac{1}{32n}.
\]
\end{enumerate}
\end{lemma}
\begin{proof}
\begin{enumerate}
    \item We extend Lemma \ref{lemma:sampling} to fresh nodes and show all token reach a node $v \in V_t$ with the same (but not necessarily uniform) probability. 
Consider a token $\theta$ independently of its source node and let $X(\theta,u)$ indicate that $\theta$ reaches $u$.
Further, let $V'_{t-1} \subset M_{t-1}$ be the set of all nodes that know $v$'s $ID$.
If $v$ receives $\theta$ in round $t$, 
then the following two events must happen
\begin{enumerate}
\item The token must be sent to any mature node $v' \in V_{t-1}$ using \sampling. 
We denote this event as $X_1(\theta,v')$.
\item Given any $v' \in V'_{t-1}$ received $\theta$, 
it must forwarded is to $v$ in round $t-1$. 
We denote this event as $X^{v'}_2(\theta, v)$.  
\end{enumerate}
We can easily show that both these events have the same probability for two tokens of possibly different origin.
The uniformity of the first event directly follows from Lemma \ref{lemma:sampling}.
Here, we showed that $\pr{X_1(w,v')} = \pr{X_1(u,v')}$ for every $u,w \in V_{t-(2\lambda+5)}$ and $v \in V_t$. Note that for two different $v',v'' \in V_t$ the probabilities $\pr{X_1(w,v')}$ and $\pr{X_1(w,v')}$ may differ.
The uniformity of the second event follows from the fact that each token is forwarded to $v$ with probability of exactly $\frac{1}{4\delta}$. 
To finalize the proof, consider two nodes $u,w \in V_{t-\lambda}$ and let 
$\theta^u$ and $\theta^w$ be tokens sent by $v$ and $w$ respectively. Then,
\begin{align*}
    \pr{X(\theta^u,v)=1} &=\sum_{v' \in V'_{t-1}}\pr{X_1(\theta^u,v')=1 \cap  X^{v^\prime}_2(\theta^u,v)=1}\\
    &= \sum_{v' \in V'_{t-1}}\pr{X_1(\theta^u,v')=1}\cdot \pr{X^{v^\prime}_2(\theta^u,v)=1 \mid X_1(\theta^u,v')=1}\\
    &= \sum_{v' \in V'_{t-1}}\pr{X_1(\theta^w,v')=1}\cdot \pr{X^{v^\prime}_2(\theta^w,v)=1 \mid X_1(\theta^w,v')=1}\\
    &= \pr{X(\theta^w,v)=1}.
\end{align*}
Here, the first equality is due to the law of total probability and second equality is due to Lemma~\ref{lemma:sampling} and the fact that each mature node $v^\prime$ forwards a token to a fresh node with probability $\frac{1}{4\delta}$. 
\item 
The fact that $\pr{X(v,w) = 1} \in \Omega(\frac{1}{n})$ then follows from three facts: 
\begin{enumerate}
    \item First, a token reaches a given mature node with probability at least $\frac{1}{4n}$. This follows directly from Lemma \ref{lemma:sampling}. 
    \item  Second, each fresh node is connected to at least $\frac{\delta}{2}$ mature nodes w.h.p. This follows because we assume that Lemma \ref{lemma:random-overlay} holds true in round $t-1$.
    \item Last, a mature node forwards a token to a connected node with probability $\frac{1}{4\delta}$.
\end{enumerate}
Combining these three facts yields the result. Formally:
\begin{align*}
    \pr{X(\theta,v)=1 \mid\big|V'_{t-1}\big| \geq \nicefrac{\delta}{2}} &=\sum_{v' \in V'_{t-1}}\pr{X_1(\theta,v')=1}\cdot \pr{X^{v^\prime}_2(\theta,v)=1 \mid X_1(\theta,v')=1}\\
    &\geq \sum_{v' \in V'_{t-1}}\frac{1}{4n}\frac{1}{4\delta} 
    \\&\geq \frac{\delta}{2} \frac{1}{4n}\frac{1}{4\delta}\\
    &\geq \frac{1}{32n}.
\end{align*}

\end{enumerate}
\end{proof}
\begin{lemma}
\label{lemma:equal3}
Assume Lemma \ref{lemma:random-overlay} held until round $t-1$. 
Further, let $X(u,v)$ denote the event that any token sent by $u$ reaches $v$ in round $t$.
Then the following statements hold:
\begin{enumerate}
\item Each node sends at least one token to node $v \in V_t$ with the same probability, i.e., $\forall u,w \in V_{t-(2\lambda+5)}$
\[
    \pr{X(u,v) = 1} = \pr{X(w,v) = 1}. 
\]
\item For each $u \in V_{t-(2\lambda+2)}$ it holds
\[
    \pr{X(u,v) = 1} \geq \frac{\tau}{33n}.
\]
\end{enumerate}
\end{lemma}
\begin{proof}
\begin{enumerate}
    \item Now, we  observe the variables $X(u,v)$ and $X(w,v)$ that denotes if any of $u$'s or $w$'s tokens reached $v$.
Recall that both $u$ and $w$ send $\tau$ tokens.
We denote these tokes as $\theta^u_1 , \ldots, \theta_\tau^u$ and $\theta^u_1 , \ldots, \theta_\tau^u$.
Let $X(\theta, v)$ be defined as in Lemma~\ref{lemma:equal2}. The probability that any of these tokens reach $v$ is given by:
\[
    \pr{\bigcup_{i = 1, \ldots, \tau} X(\theta^u_i,v) = 1} = 1-\pr{\bigcap_{i=1, \ldots, \tau}  X(\theta^u_i,v)=0}
\]
Since all these tokens are independent, it holds that:
\[
     \pr{\bigcap_{i = 1, \ldots, \tau}  X(\theta^u_i,v)=0} = \prod_{i = 1, \ldots, \tau} \left(1-\pr{X(\theta^u_i,v)=1}\right)
\]
The same holds respectively for $\pr{\bigcap_{i = 1, \ldots, \tau} X(\theta^w_i,v) = 0}$.
Putting these observations together, we get that:
\begin{align*}
    \pr{X(u,v) = 1} &= \pr{\bigcup_{i = 1, \ldots, \tau} X(\theta^u_i,v) = 1} = 1-\pr{\bigcap_{i=1, \ldots, \tau}  X(\theta^u_i,v)=0}\\
    &= 1-\prod_{i=1, \ldots, \tau}\left(1-\pr{X(\theta^u_i,v)=1}\right)\\
    &= 1-\prod_{i=1, \ldots, \tau}\left(1-\pr{X(\theta^w_i,v)=1}\right)\\
    &= 1-\pr{\bigcap_{i=1, \ldots, \tau}  X(\theta^w_i,v)=0} = \pr{\bigcup_{i = 1, \ldots, \tau} X(\theta^w_i,v) = 1}\\
    &= \pr{X(w,v) = 1}.
\end{align*}
This was to be shown.

\item 
For any pair of $v$ and $u$, the probability is lower bounded by 
\begin{align*}
    \pr{X(u,v)=1} &= 1 - \pr{X(u,v)=0}\\
    &= 1 - \prod_{i = 1, \ldots, \tau} \left(1-\pr{X(\theta_i^u,v)=1}\right)\\
    &= 1 - \left(1-\frac{1}{32n}\right)^{\tau}\\
    & \geq 1-\exp\left(-\frac{\tau}{32n}\right)\\
    & \geq 1-\left(1-\frac{\tau}{32n} + \left(\frac{\tau}{32n}\right)^2\right)\\
    & \geq \frac{\tau}{32n} - \left(\frac{\tau}{32n}\right)^2\\
    & \geq \frac{\tau}{33n}\\
\end{align*}
where for the second inequality we use the fact that for all $x \leq 1$,
\[\exp(x) \leq 1+x+x^2.\]
The last inequality follows from fact that $\tau \in O(\log n)$ and thus $\frac{\tau}{32n}$ can be made arbitrarily small for a big enough $n$. 
\end{enumerate}
\end{proof}

Together with our assumptions on the churn rate, we get the following lemma:

\begin{lemma}
\label{lemma:distinct_token}
Let $c \geq 280k$.
Let each mature node start $\tau \geq 20c\lambda$ token, then each fresh node receives at least $\frac{\tau}{100}$ distinct token with probability at least $1-\frac{1}{n^k}$.
\end{lemma}
\begin{proof}
Recall that at least $\frac{15}{16}n$ mature nodes in round $t-(2\lambda+5)$ that start $\tau$ tokens each.
Hence, the minimal number of nodes that start tokens is at least $K :=\frac{15 n}{16}$.
Further, there are at most $\frac{17}{16}n$ nodes in round $t$.
Fix a node $v$ and let $X_1, \ldots, X_{K}$ be the indicator variables that a nodes has a token that reaches $v$.
Then the expected number of distinct tokens received by node $v$ is given by,
\begin{align*}
    E\left[\sum_{i=1}^{K} X_i\right] \geq \sum_{i=1}^{K} \frac{\tau}{33n} \geq \frac{15\tau n}{16\cdot33\cdot n} \geq \frac{\tau}{50},
\end{align*}
where we use $\pr{X_i = 1} = \frac{\tau}{33n}$ due to Lemma~\ref{lemma:equal3}.
Given that all mature nodes send their tokens independent of one another, the Chernoff Bound is applicable and the lemma follows for a big enough $c$. 
In particular, it holds for $c \geq 280k$:
\begin{align*}
    \pr{X \leq \frac{\tau}{100}} &= \pr{X \leq (1-\nicefrac{1}{2})\frac{\tau}{50}}\\
    &\leq \exp\left(-\frac{\tau}{4 \cdot 50 \cdot 3}\right) \\
    &\leq \exp\left(-k\lambda\right) = n^{-k}.
\end{align*}
\end{proof}

This basically tells us that - as long as we choose $\tau$ bigger than $20c\lambda$ - each node will receive roughly $\Omega(\tau)$ distinct tokens w.h.p, which it can then use to advertise itself and the new nodes connected to it.

Next, we need to consider, how big we need to choose $\tau$ such that each node has enough tokens to ensure that it is able to connect to $\frac{\delta}{2}$ mature nodes.
Lemma~\ref{lemma:distinct_token} gives us that choosing $\tau$ such that $\frac{\tau}{100} \ge (\phi+1)\cdot \delta$, where $\phi = O(1)$ is the maximum number of nodes that could join via a fresh node in any given round, then each fresh node has enough tokens to advertise itself to $\delta$ distinct mature nodes every round and also provide for the newly joined nodes.

In the following we can assume that each fresh node sends a connection request to $\delta$ nodes.
However, these requests can still fail for two reasons: 
\begin{enumerate}
    \item First, the $ID$ of the token used for the connections belongs to a node that has been churned out.
    \item Second, the the target has received more than $2\delta$ connection requests and refuses the connection.
\end{enumerate}
The first factor depends on the number of nodes have been churned out and on the numbers of connections we make. 
The second term only depends on the random process that creates these edges.

We begin by showing that only a small fraction of connection request are sent to churned out nodes. In the following lemma we say node $v$ is good in round $t$ if and only if $v \in V_{t-(2\lambda + 5)} \cap V_{t+2}$ and referred to as bad, otherwise. 

\begin{lemma}
\label{lemma:bad_tokens}
Suppose that $\tau \geq 26000k\lambda$ and $\delta\geq 60k\lambda$, then
each fresh node has at least $\frac{\delta}{2}$ connections to good nodes with probability at least $1-\frac{1}{n^k}$.
\end{lemma}
\begin{proof}
The proof of this lemma is straightforward and mostly technical. 
The basic outline is as follows:
Due to its lateness the adversary cannot anticipate where a node will send its tokens. Thus, the tokens of good and bad nodes will \emph{randomly} spread to the fresh nodes.
As will see, in expectation each fresh node roughly receives a $\frac{15}{16}$ fraction of good nodes. Since the sampling is independent, this implies there is at least a $\frac{13}{18}$ fraction \emph{w.h.p} (for a big enough $\tau$) due to the Chernoff Bound.  
Since a fresh node randomly draws its connection without replacement there are also $\frac{13}{18}\delta$ successful connections in expectation. Since drawing without replacement is NA, another application of the Chernoff bound concludes the proof.

We will now prove these claims in more detail:
Fix a node $v \in V_t$ that advertises itself or a newly joined node.
Let $X_1, \dots, X_{\delta}$ be the binary RVs such that $X_j$ denotes if $j^{th}$ advertisement by $v$ is successful, i.e., its identifier is advertised to a good node.
The outcome of $X = \sum_i X_i$ depends on two values, the overall number of tokens that $v$ received and the number of identifiers of good nodes.

We have already established that the number of distinct tokens that a node receives can be subjected to the Chernoff Bound and is therefore concentrated around its expectation.
The same holds for the number of good identifiers.
Let $G$  be the number of good identifiers that $v$ draws from the set of available tokens.
One can easily verify that $G$ is the sum independent binary random variables: For each good node $w \in V_{t-(2\lambda+5) \cap V_{t+2}}$, let $G_w \in \{0,1\}$ be the indicator for the event that $w$ sends one token with its identifier to $v$. 
Then it holds $G := \sum_{w \in V_{t-(2\lambda+5)} \cap V_{t+2}} G_w$ and all $G_w$'s are independent. 

Recall that at least $n\left( 1 - \frac{1}{16}\right)$ and at most $n\left( 1 + \frac{1}{16}\right)$ nodes started tokens $2\lambda+5$ rounds ago.
Since at most a $\nicefrac{1}{16}$-fraction of all nodes that started a token are churned out until round $t+2$, it holds that $\left(\frac{15}{16}\right)^2n$ is a lower bound for the number of good tokens started $2\lambda+5$ rounds ago.

Let now $Y$ be the number of all distinct tokens that $v$ received.
Let $p \in [\frac{\tau}{33n},\frac{\tau}{n}]$ be the probability that at least token of a fixed node reaches $v$.
Then, it holds $E[Y] \le pn\left( 1 + \frac{1}{16}\right)$ and $E[G] \ge p\left(\frac{15}{16}\right)^2n$.
We will upper and lower bound $Y$ and $G$ respectively.
We start with $Y$. 
Assuming that $p n \geq 768k\lambda$ the Chernoff Bound gives us that,
\begin{align*}
    \pr{Y \ge \frac{9}{8} pn} &= \pr{Y \geq (1+\nicefrac{1}{17})(1+\nicefrac{1}{16})pn} 
    \leq \exp\left(-\frac{pn}{2 \cdot 17 \cdot 16}\right) 
    \leq \exp\left(-k\lambda\right) = n^{-k}.
    \label{equ:token} \tag{9}
\end{align*}
\begin{align*}
    \pr{G \leq \frac{13}{16}pn} \le \pr{G \le (1-\nicefrac{1}{15})\left(\frac{15}{16}\right)^2pn} &\leq \exp\left(-\frac{15^2\cdot pn}{3 \cdot 16^2 \cdot 15^2 }\right) \leq
    \exp\left(-k\lambda\right) = n^{-k}.
    \label{equ:good_token} \tag{10}
\end{align*}
Therefore, it remains to show that we can choose $pn$ big enough for these statements to hold.
Recall that $\frac{\tau}{n} \ge p \geq \frac{\tau}{33n}$,
then for $\tau \ge 26000 k\lambda$, we have that $pn \ge 768k\lambda$.

Now, we condition on (\ref{equ:token}) and (\ref{equ:good_token}) being false and denote this event as $\mathcal{G}$.
In this case, a simple calculation reveals that at least a $\frac{13}{18}$-fraction of tokens is good. 
If we pick $\delta' = \min\{\delta,Y\}$ of these tokens uniformly at random without replacement, a constant fraction will point to good nodes in expectation:
\begin{align*}
    E[X \mid \mathcal{G}] 
   &\ge \sum_{g=\frac{13}{16}pn}^{n\left(1 + \frac{1}{16}\right)}\sum_{y=0}^{\frac{9}{8}pn} \pr{G=g, Y=y \mid \mathcal{G}} \cdot \delta' \frac{g}{y}\\
   &\geq \sum_{g=\frac{13}{16}pn}^{n\left(1 + \frac{1}{16}\right)}\sum_{y=0}^{\frac{9}{8}pn} \pr{G=g, Y=y \mid \mathcal{G}} \cdot \delta'\frac{\frac{13}{16}pn}{\frac{9}{8}pn}\\ 
   &= \delta'\frac{\frac{13}{16}pn}{\frac{9}{8}pn} \sum_{g=\frac{13}{16}pn}^{n\left(1 + \frac{1}{16}\right)}\sum_{y=0}^{\frac{9}{8}pn} \pr{G=g, Y=y \mid \mathcal{G}} \\ 
    &= \delta'\frac{\frac{13}{16}pn}{\frac{9}{8}pn}\\ 
    &= \delta'\frac{13}{18}.
\end{align*}
We show that under these circumstances at least half of all advertisements go to good nodes for a big enough $\delta'$.
Note that we observe a hyper-geometric distribution, which is known to be NA (cf. \cite{dubhashi1998balls}).
Thus, by the Chernoff Bound, a constant fraction points to living nodes \emph{w.h.p.} if we choose $\delta'$ high enough.
In particular, by choosing $\delta \geq 60k\lambda$
\begin{align*}
    \pr{\sum_{i=0}^{\delta} X_i \leq \frac{1}{2}\delta \mid \mathcal{G}} &\leq \pr{\sum_{i=0}^{\delta} X_i \leq (1-\nicefrac{5}{18})\frac{13}{18}\delta}\\
    &\leq \exp\left(-\frac{5^2\cdot13\delta}{3 \cdot 18^3}\right) \leq \exp\left(-\frac{\delta}{60}\right) & \rhd\textit{Using }\frac{5^2 \cdot 13}{3\cdot 18^3} \geq \frac{1}{60}\\
    &\leq \exp\left(-k\lambda\right) = n^{-k}. & \rhd\textit{for }\delta \geq 60k\lambda\\
\end{align*}
Note that since $\mathcal{G}$ holds w.h.p, we have that 
\begin{align*}
    \pr{\sum_{i=0}^{\delta} X_i \leq \frac{1}{2}\delta} &= \pr{\mathcal{G}}\cdot \pr{\sum_{i=0}^{\delta} X_i \leq \frac{1}{2}\delta \mid \mathcal{G}} + \pr{\neg \mathcal{G}}\cdot \pr{\sum_{i=0}^{\delta} X_i \leq \frac{1}{2}\delta \mid \neg\mathcal{G}}\\
    &\leq \pr{\mathcal{G}}\cdot \pr{\sum_{i=0}^{\delta} X_i \leq \frac{1}{2}\delta \mid \mathcal{G}} + \pr{\neg \mathcal{G}} \cdot 1\\
    &= \pr{\mathcal{G}}\cdot \pr{\sum_{i=0}^{\delta} X_i \leq \frac{1}{2}\delta \mid \mathcal{G}} + \pr{\left(Y \geq (1+\nicefrac{1}{15})pn\right) \cup \left(|G| \leq \frac{14}{16}pn\right)}\\
    &\leq \left(1-\frac{1}{n^k}\right)\frac{1}{n^k} + \frac{2}{n^k} \leq \frac{3}{n^k}. 
\end{align*}
Thus, the statement holds w.h.p.
Finally, note that our choice of $\tau$ already implies that we receive that least $\delta$ tokens.
\end{proof}
\begin{lemma}
\label{lemma:oblivous_random}
A $(2, 2\lambda+7)$-late adversary enables \random ensure that every fresh node is connected to $\frac{\delta}{2}$ mature nodes in each round.
\end{lemma}
\begin{proof}
The proof follows using the correctness of \routing and \random.
Note that the adversary is oblivious of the random edges because they only persist for $2$ rounds. 
Each mature node disseminates tokens to random positions in the $[0,1)$ interval. The tokens arrive at their target node for being sampled after $2\lambda+2$ rounds. The mature nodes that receive the token forward them to fresh nodes which in turn connect to the mature nodes to stay connected in the network until they mature themselves. The fresh nodes then receive new tokens from these connections. The entire process takes $2\lambda+5$ rounds in total. Therefore, in any given round $t$ a $(2, 2\lambda+7)$ adversary is oblivious to any communication between the fresh nodes and mature nodes, since all connections established until round $t$ are already defunct, i.e., the adversary is unable to anticipate which tokens reach a given fresh node. This in turn enables \random maintain the invariant every round.
\end{proof}

Lemma~\ref{lemma:oblivous_random} and the churn parameters ensure that there exists a constant size set of good nodes $G_t := V_{t-(2\lambda + 5)} \cap V_{t+2}$ that send a token in round $t-(2\lambda + 5)$ and are not churned out until round $V_{t+2}$. Therefore, if a node receives enough tokens of good nodes, it can successfully advertise its identifier w.h.p.

It remains to show that at most $2\delta$ fresh nodes connect to a mature node in any given round. We first analyze the expected number of incoming connections.

\begin{lemma}\label{lemma:equal4}
Let $A(u,v)$ denote the event that $v$ advertises itself to $u$.
Then, $\forall u,w \in V_t$ and any two tokens $\theta,\theta'$ it holds that, 
\begin{align*}
    \pr{A(u,v)} = \pr{A(w,v)}. 
\end{align*}
\end{lemma}
\begin{proof}
We again divide the process into two stages.
First, the token of $u$ needs to reach $v$ and then needs to be picked for the advertisement.
Let $Y_1(u,v)$ and $Y_1(w,v)$ the respective events that tokens of $u$ and $w$ reached $v$.
By Lemma \ref{lemma:equal2} we already know that these events have the same probability.

Note that the actual choice of the nodes which are advertised only depends on the number of distinct available tokens.
In particular, given that a node received $\ell$ distinct tokens, the probability for one of these tokens to be used is $\min\{1, \frac{\delta}{\ell}\}$. 
This follows from the fact that we draw (up to) $\delta$ tokens uniformly at random or all tokens if we received less than $\delta$. Thus, we draw without replacements and observe a hyper geometric distribution.

Let now $N_v$ be the number of distinct tokens received by $v$.
Since all nodes send (at least) one token to $v$ independently and with same probability $p := \pr{Y_1(u,v)=1}$, the value of $N_v$ only depends on the number of nodes we observe.
More formally, it holds:
\begin{align*}
    \pr{N_v = \ell} &:= \sum_{S \subset M_t, |S| = \ell} \pr{\bigcap_{u \in S}Y_1(u,v)=1}\cdot \pr{\bigcap_{u' \not\in S}Y_1(u',v)=0}\\
    &:= \sum_{S \subset M_t, |S| = \ell} \prod_{u \in S}\pr{Y_1(u,v)=1}\cdot \prod_{u' \not\in S}\pr{Y_1(u',v)=0}\\
    &= \binom{|M_t|}{\ell}p^\ell(1-p)^{n-\ell}.
\end{align*}
Thus, if we condition on the fact that $v$ already received a token of a certain node, the probability that this node receives $\ell-1$ additional tokens from different nodes stays the same.
It holds,
\begin{align*}
    \pr{N_v = \ell \mid Y_1(u,v)=1} &=  \sum_{S \subset M_t\setminus\{u\}: |S| = \ell-1} \pr{\bigcap_{u^\prime \in S}Y_1(u^\prime,v)=1}\cdot \pr{\bigcap_{u^{\prime\prime} \not\in S}Y_1(u^{\prime\prime},v)=0}\\
    &= \binom{|M_t|-1}{\ell-1} \cdot p^{\ell-1} \cdot (1-p)^{|M_t|-(\ell-1)}\\
    &= \sum_{S \subset V\setminus\{w\}: |S| = \ell-1}\pr{\bigcap_{u^\prime \in S}Y_1(u^\prime,v)=1}\cdot \pr{\bigcap_{u^{\prime\prime} \not\in S}Y_1(u^{\prime\prime},v)=0}\\
    &= \pr{N_v = \ell \mid Y_1(w,v)=1}.
\end{align*}
Thus, when summing over all different outcomes, we get,
\begin{align*}
    \pr{A(u,v)=1} &= \pr{Y_1(u,v)=1} \cdot \pr{A(u,v) \mid Y_1(u,v)=1}\\
    &= \pr{Y_1(u,v)=1} \cdot \left(\sum_{\ell=1}^{|M_t|}\pr{N_v = \ell \mid Y_1(u,v)}\cdot \min\left\{1, \frac{\delta}{\ell}\right\}\right)\\
    &= \pr{Y_1(w,v)=1} \cdot \left(\sum_{\ell=1}^{|M_t|}\pr{N_v = \ell\mid Y_1(w,v)=1}\cdot \min\left\{1, \frac{\delta}{\ell}\right\}\right)\\
     &= \pr{Y_1(w,v)=1} \cdot \pr{A(w,v) \mid Y_1(w,v)=1}\\
    &= \pr{A(w,v)=1}.
\end{align*}
This was to be shown.
\end{proof}

Intuitively, this lemma implies that all node have the same probability of receiving an advertisement.
Thus, since there are at $\delta n$ advertisement, the expected number of incoming connections is bound by $\delta$.

\begin{lemma}\label{lemma_exp_adv}
Fix a mature node $w$ that started tokens $2\lambda+5$ rounds ago.
Let $X$ be a random variable that denotes the number of fresh nodes that advertise themselves to $w$.
It holds:
\begin{align*}
    \E{X} \leq \delta.
\end{align*}
\end{lemma}
\begin{proof}
Let $f := |F_t|$ be number of fresh nodes and $m = |M_t|$ the number of mature nodes in round $t$. Note that $\frac{f}{m}$ is at most $\frac{1}{8}$ due to 
our choice of $\alpha$ and $\kappa$.

For a fixed $v \in F_t$ and $w\in M_t$ let $A(v,w)$ be the binary RV
that denotes if $v$ connects to $w$. 
Let now $A_v :=\sum_{w \in M_t} A(v,w) $ be the random variable that counts $v$'s advertisements. 
Since each node creates at most $\delta$ advertisements, it must hold that:
\[
    \E{A_v} \leq \delta
\]
Further, we have that
\[
    \E{A_v} = \sum_{w \in V} E[A(v,w)]
\]
This follows from the linearity of expectation.
Given that $A(v,w)$ is a binary random variable, we also have that
\[
    \E{A(v,w)} = \pr{A(v,w)}
\]
Using Lemma \ref{lemma:equal4} we also know that $\pr{A(u,v)}=\pr{A(w,v)}$.
Now we combine our observations to bound $\pr{A(u,v)}$.
First, see that:
\begin{align*}
    \delta \geq \E{A_v} &= \sum_{w \in V} \E{A(v,w)} & \rhd \textit{By linearity of Expectation}\\
    &= \sum_{w \in V} \pr{A(v,w)} & \rhd \textit{As } \E{A(v,w)} = \pr{A(v,w)}\\
    &= m \cdot \pr{A(v,w)} & \rhd \textit{As } \pr{A(v,w)} = \pr{A(v,w')}\\
\end{align*}    
Therefore,
\begin{align*}
    \pr{A(v,w)} \leq \frac{\delta}{m}.
\end{align*}
Let $X$ be a random variable that counts the number of incoming connections to the node $w$. Then, 
\[
    \E{X} = \sum_{v \in F_t} \E{A(v,w)} \leq \frac{\delta}{8}.
\]
Thus, our claim holds.
\end{proof}
\begin{lemma}\label{lemma_maxadv}
Each mature node receives at most $2\delta$ connections from fresh nodes w.h.p.
\end{lemma}
\begin{proof}
Fix a mature node $w \in M_t$ and let $X_{v} \in \{0,1\}$ be the random variable that indicates whether $v \in F_t$ connects to $w$. 
Further, let $X := \sum_{v \in F_t} X_v$ the sum of all fresh nodes that connect to $w$.

We will show that $(X_v)_{v \in F_t}$ is negatively correlated and then use the Chernoff Bound on $X := \sum_{v \in F_t} X_v$ to prove the lemma. 

Let $A$ and $B$ be two disjoint subsets of fresh nodes and $X_A$ and $X_B$ the respective subsets of variables that correspond to these nodes.
We show that,
\[
     \E{X_A \cdot X_B} \leq \E{X_A} \cdot \E{X_B}.
\]
A simple induction then implies that $(X)_{v \in F_t}$ is negatively correlated:
Pick any set $S \subset [1,n]$ and denote the indices in $S$ as $i_1, \dots, i_s$. 
Given that the formula above is true for any disjoint subset, we have:
\begin{align*}
    \E{\prod_{j = i_1, \ldots, i_s} X_j} &\leq \E{X_{i_1}} \cdot \E{\prod_{j = i_2, \ldots, i_s} X_j}\\
    &\leq \E{X_{i_1}} \cdot \E{X_{i_2}} \cdot \E{\prod_{j = i_3, \ldots, i_s} X_j}\\
    &\leq \hspace{2cm}\vdots\\
    &= \prod_{j = i_1, \ldots, i_s} \E{X_j}. 
\end{align*}
Here, the inequalities hold by choosing $A = \{i_1\}$ and $B=S \setminus A$ in the first line.

Thus, we will now show that
\begin{align*}
     \E{X_A \cdot X_B} \leq \E{X_A} \cdot \E{X_B}.
\end{align*}
Observe that,
\begin{align*}
     &\E{X_A X_B} \leq \E{X_A} \cdot \E{X_B}
     \\\iff &\E{X_A X_B} - \E{X_A} \cdot \E{X_B} \leq 0 \\
     \iff &Cov(X_A,X_B) \leq 0. 
\end{align*}
Thus, we will show that $Cov(X_A,X_B) \leq 0$.

If we condition on the set of tokens that each node receives, then $(X_v)_{v \in F_t}$ (and therefore $X_A$ and $X_B$) follows the hypergeometric distribution as we draw the received tokens without replacement.
Thus, intuitively for any fixed distribution of tokens to fresh nodes, the conditioned distribution is therefore negatively correlated.
To show this more formally, we need to to precisely condition $(X_v)_{v \in F_t}$ on the distribution of tokens.
For a simpler presentations denote all tokens sent by $w$ as \emph{red} tokens, all others are \emph{blue} in the following.
Let $R_v \in \{0,1\}$ the random variable that indicates that $v \in F_t$ received red token.
Further let $B_{v}$ count the number of distinct blue tokens \emph{not} received by $v$.
Thus, the set  $Y := \left(R_v \cup B_v\right)_{v\in F_t}$ completely characterizes the distribution of tokens to nodes.

Using the closure properties of NA, one can show that the set $Y := \left(R_v \cup B_v\right)_{v\in F_t}$ is NA: 
For each each red token $\theta$ we can define the binary variables $X(\theta, v)$.
It holds that $X(\theta,v) = 1$ if $\theta$ reaches $v$, and $0$ otherwise.
Observe that for a fixed $\theta$ there is at most one $X(\theta,v) = 1$ and all others are $0$ as a token can only be received by one node.
Thus, for each token $\theta$ the set $\left(X(\theta,v)\right)_{v \in V_t}$ is NA.
Further, if we fix a node $v$ the sum of all $X(\theta,v)$ is NA as well.
The variable $R_v$ is now simply defined as 
\begin{align*}
    R_v := \begin{cases}
      1 & \textit{if } \sum_{i=1}^\tau X(\theta_i^w,v) > 0\\
      0 & \textit{else }
    \end{cases}
\end{align*}
Since it is monotonically increasing in $\sum_{i=1}^\tau X(\theta_i^w,v)$ the variable $R_v$ is also NA.

The same holds almost analogously for the blue tokens, however we need to adapt to the fact that we count the identifiers that were not received.
Therefore, we define $\overline{X}(\theta,v) := 1 - X(\theta,v)$.
In particular, we set $\overline{X}(\theta,v) = 0$ if $\theta$ reaches $v$, and $1$ otherwise.
As $X(\theta,v)$ is NA, so is $\overline{X}(\theta,v)$ as it is monotone function.
For each node $u \in V\setminus\{w\}$ now define $B^u_v$ to be binary variable indicating if any token of $u$ reached $v$.
Thus, $B^u_v$ equals $1$ if all $\tau$ tokens of $u$ missed $v$. 
We formalize this by setting:
\begin{align*}
    B_v^u := \begin{cases}
      1 & \textit{if } \sum_{i=1}^\tau \overline{X}(\theta^u,v) = \tau\\
      0 & \textit{else }
    \end{cases}
\end{align*}
Again, $\sum_{i=1}^\tau \overline{X}(\theta^u_i,v)$ is NA and thus $B^u_v$ is increasing function of $\sum_{i=1}^\tau \overline{X}(\theta_i^u, v)$ it is NA, too.
By summing over all $B^u_v$, we obtain the number of distinct tokens that did not reach $v$. Further, the variables $(B_v)_{v \in F_t}$ with $B_v := \sum_{u \in V} B^u_v$ are NA as each variable is a monotone function on disjoint NA variables. 

Now observe that all $X_v$'s are independent of one another given $Y$.
Further, the expected value $E[X_v|R_v,B_v]$ 
of each $X_v$ is dependent solely on the variables $R_v$ and $B_v$. 
In particular, each $E[X_v|R_v,B_v]$ monotonically rises in both $R_v$ and $B_v$ and is given by,
\begin{align*}
    \E{X_v \mid Y = (R_v, B_v)} &= \pr{X_v \mid Y = (R_v, B_v)}
    \\&= \frac{R_v}{R_v + (m - B_v)},
\end{align*}
where $m = |V\setminus\{w\}|$ is the total number of nodes in round $t - (2\lambda + 5)$ that send blue tokens.

Obviously, receiving a red tokens raises the expectation of drawing a red token. However, it also holds that, the less blue token (i.e., more blue tokens that we don't receive) we receive the more likely it becomes to draw a red token.
Thus, given two disjoint subsets $X_A, X_B \subset X$ we can view $E[X_A|Y]$ and $E[X_B|Y]$ as functions that monotonically rise in disjoint subsets $Y_A = (R_A \cup B_A)$ and $Y_B := (R_B \cup B_B)$.
With these observations, we can now show that $X$ is negatively correlated using only some technical arguments.
By the law of total convariance, it holds: 
\[
	Cov\left((X_A,X_B\right) := \E{Cov(X_A,X_B)|Y} + Cov\left(\E{X_A|Y},\E{X_B|Y}\right) 
\]
We see that the first term is $0$ since $(X_v)_{v \in F_t}$ is independent given $Y$.
Note that the covariance of independent variables is always $0$ by definition. 
Thus, it holds:
\[
	Cov\left(X_A,X_B\right) = Cov\left(\E{X_A|Y},\E{X_B|Y}\right) 
\]
It remains to show that this term is smaller than $0$. 
First, recall that $E[X_A|Y]$ and $E[X_B|Y]$ are 
monotonically increasing in $Y_A$ and $Y_B$.
Thus, we can view them as monotonically increasing functions $f(Y_A) := \E{X_A|Y}$ and $g(Y_B) := \E{X_B|Y}$ for disjoint subsets of $Y$. 
Further, we showed $Y$ is NA and thus - by the closure properties of NA - its holds that $f(Y_A)$ and $f(Y_B)$ are NA, too.
Therefore, the fact that $Cov(f(Y_A),g(Y_B)) \leq 0$ follows from the definition of NA.
And thus: 
\[
Cov\left(\E{X_A|Y},\E{X_B|Y}\right) = Cov\left(f(Y_A),g(Y_B)\right) \leq 0.
\]

This implies that $(X_v)_{v \in F_t}$ is negatively correlated. 

Now, we need one last application of the Chernoff Bound.
Using Lemma~\ref{lemma_exp_adv}, we know that $E[X] \leq \delta$.
Now observe that, 
\[\pr{X > 2\delta} = \pr{X > 2\frac{\delta}{E[X]}E[X]}.\]
Note that $\frac{\delta}{E[X]} \geq 1$ and we can therefore use the third bound from Lemma \ref{chernoff_bound} for parameters bigger than $1$.
Thus, for $\delta \geq 3k\lambda$ we get that:
\begin{align*}
    \pr{X \geq 2\delta} &\leq \exp\left(-\frac{\delta}{3}\right)\\ 
    &\leq \exp\left(-k\lambda\right) = \frac{1}{n^k}.
\end{align*}
This shows the lemma as was claimed.
\end{proof}


\subsection{Congestion}
\begin{lemma}
\label{lemma:congestion}
Algorithms \maintainer and \random have congestion of $O(\log^3 n)$ per node and round w.h.p.
\end{lemma}
\begin{proof}
We observe the number of messages due to \maintainer and \random by invocation of \routing. 
We observe the two subroutines \maintainer and \random separately.
\begin{enumerate}
\item In \maintainer each round every mature starts three routing requests for itself and three routing requests on behalf of each fresh node connected to it.
Since there are at most $2\delta$ fresh nodes connected to a mature node
\emph{w.h.p.}, a given mature node starts $O(\log n)$ routing requests.
\item In \random each round every mature starts $O(\log n)$ tokens per round. Each token corresponds to one routing request.
\end{enumerate}
Using Lemma~\ref{lemma:routing_correctness} that each routing takes $O(\log n)$ rounds and Lemma~\ref{lemma:interval_congestion2} that for each routing request \routing has a congestion of $O(\log n)$ per round, \maintainer and \random together have congestion $O(\log^3 n)$ per round.

Now we observe the remaining operations performed each round.
\begin{enumerate}
    \item Recall that each swarm is of size $O(\log n)$ w.h.p. 
    Thus, during the introduction step in \maintainer each mature node introduces $O(\log n)$ nodes to their $O(\log n)$ neighbors. 
    Resulting in a congestion of $O(\log^2 n)$ additional messages per node. 
    \item In \random each node, w.h.p., receives $O(\log n)$ tokens through the sampling algorithm and forwards them to fresh nodes. Additionally, each fresh node sends out $O(\log n)$ advertisements.
    Thus, altogether each node exchanges $O(\log n)$ messages.
\end{enumerate}
\end{proof}


Theorem 1 now follows from lemmas \ref{lemma:correctness_lds}, \ref{lemma:random-overlay} and \ref{lemma:congestion}.

\section{Future Work \& Conclusion}
We presented an algorithm that maintains a structured overlay in presence of a $(2, O(\log n))$-late adversary. We permit $\alpha n$ deletions/additions over the course of $O(\log n)$ rounds. Note that this is exponentially higher than in \cite{AugustineS18} and \cite{DreesGS16}. 
However, both their algorithms are \emph{not} possible if the adversary has more recent knowledge of topology. This suggests a strong connection between an adversaries lateness with regard to the topology and permitted churn. For future work, one could consider finding an algorithm that tolerates
a $(1, O(\log n))$-late adversary. Also one could consider a hybrid model where the adversary has 
an almost up-to-date information about some nodes but is more outdated with regard to others. 

Further, we did not consider any kind of byzantine behavior.
However, the approaches used by Fiat et. al. \cite{FiatSY05} could perhaps also be used with our overlay.
Given, the overlay can handle byzantine nodes, further overlay problem, i.e. distributed agreement in the $(a,b)$-late setting, could also promising directions for future work.

\bibliographystyle{plain}
\bibliography{references}

\begin{thebibliography}{10}

\bibitem{AugustineMMPRU13}
John Augustine, Anisur~Rahaman Molla, Ehab Morsy, Gopal Pandurangan, Peter
  Robinson, and Eli Upfal.
\newblock Storage and search in dynamic peer-to-peer networks.
\newblock In {\em Proc. of {SPAA}}, pages 53--62, 2013.

\bibitem{AugustineP0RU15}
John Augustine, Gopal Pandurangan, Peter Robinson, Scott~T. Roche, and Eli
  Upfal.
\newblock Enabling robust and efficient distributed computation in dynamic
  peer-to-peer networks.
\newblock In {\em Proc. of {FOCS}}, pages 350--369, 2015.

\bibitem{AugustineS18}
John Augustine and Sumathi Sivasubramaniam.
\newblock Spartan: {A} framework for sparse robust addressable networks.
\newblock In {\em Proc. of {IPDPS}}, pages 1060--1069, 2018.

\bibitem{AwerbuchS07}
Baruch Awerbuch and Christian Scheideler.
\newblock Towards scalable and robust overlay networks.
\newblock In {\em Proc. of {IPTPS}}, 2007.

\bibitem{bellare1993random}
Mihir Bellare and Phillip Rogaway.
\newblock Random oracles are practical: A paradigm for designing efficient
  protocols.
\newblock In {\em Proceedings of the 1st ACM Conference on Computer and
  Communications Security}, pages 62--73, 1993.

\bibitem{DreesGS16}
Maximilian Drees, Robert Gmyr, and Christian Scheideler.
\newblock Churn- and dos-resistant overlay networks based on network
  reconfiguration.
\newblock In {\em Proc. of {SPAA}}, pages 417--427, 2016.

\bibitem{dubhashi1998balls}
Devdatt Dubhashi and Desh Ranjan.
\newblock Balls and bins: A study in negative dependence.
\newblock {\em Random Structures \& Algorithms}, 13(2):99--124, 1998.

\bibitem{DBLP:conf/sss/FeldmannS17}
Michael Feldmann and Christian Scheideler.
\newblock A self-stabilizing general de bruijn graph.
\newblock In {\em Proc. of {SSS}}, pages 250--264, 2017.

\bibitem{FiatSY05}
Amos Fiat, Jared Saia, and Maxwell Young.
\newblock Making chord robust to byzantine attacks.
\newblock In {\em Proc. of {ESA}}, pages 803--814, 2005.

\bibitem{GmyrHSS17}
Robert Gmyr, Kristian Hinnenthal, Christian Scheideler, and Christian Sohler.
\newblock Distributed monitoring of network properties: The power of hybrid
  networks.
\newblock In {\em Proc. of {ICALP}}, 2017.

\bibitem{GoetteHSW20}
Thorsten G{\"{o}}tte, Kristian Hinnenthal, Christian Scheideler, and Julian
  Werthmann.
\newblock Time-optimal construction of overlay networks.
\newblock {\em CoRR}, abs/2009.03987, 2020.

\bibitem{joag-dev1983}
Kumar Joag-Dev and Frank Proschan.
\newblock Negative association of random variables with applications.
\newblock {\em Ann. Statist.}, 11(1):286--295, 03 1983.

\bibitem{KingLSY07}
Valerie King, Scott Lewis, Jared Saia, and Maxwell Young.
\newblock Choosing a random peer in chord.
\newblock {\em Algorithmica}, 49(2):147--169, 2007.

\bibitem{KingS04}
Valerie King and Jared Saia.
\newblock Choosing a random peer.
\newblock In {\em Proc. of {PODC}}, pages 125--130, 2004.

\bibitem{MU05}
Michael Mitzenmacher and Eli Upfal.
\newblock {\em Probability and Computing}.
\newblock Cambridge University Press, 2005.

\bibitem{DBLP:conf/spaa/NaorW03}
Moni Naor and Udi Wieder.
\newblock Novel architectures for {P2P} applications: the continuous-discrete
  approach.
\newblock In {\em Proc. of {SPAA}}, pages 50--59, 2003.

\bibitem{DBLP:conf/sss/RichaSS11}
Andr{\'{e}}a~W. Richa, Christian Scheideler, and Phillip Stevens.
\newblock Self-stabilizing de bruijn networks.
\newblock In {\em Proc. of {SSS}}, pages 416--430, 2011.

\bibitem{ScheidelerHabil}
Christian Scheideler.
\newblock {\em Probabilistic Methods for Coordination Problems}.
\newblock PhD thesis, Paderborn University, 2000.

\bibitem{Scheideler05}
Christian Scheideler.
\newblock How to spread adversarial nodes?: rotate!
\newblock In {\em Proc. of {STOC}}, 2005.

\bibitem{Stutzbach}
Daniel Stutzbach and Reza Rejaie.
\newblock Understanding churn in peer-to-peer networks.
\newblock In {\em Proc. of SIGCOMM}, pages 189--202, 2006.

\bibitem{Wajc2017NegativeA}
David Wajc.
\newblock Negative association-definition , properties , and applications.
\newblock 2017.

\end{thebibliography}

\end{document}